\documentclass[a4paper,UKenglish]{lipics-v2018}
\usepackage{microtype}%if unwanted, comment out or use option "draft"

\title{Parameterized Complexity of Independent Set in H-Free Graphs}

\author{\'{E}douard Bonnet}{Univ Lyon, CNRS, ENS de Lyon, Université Claude Bernard Lyon 1, LIP UMR5668, France}{}{}{}

\author{Nicolas Bousquet}{CNRS, G-SCOP laboratory, Grenoble-INP, France}{}{}{}

\author{Pierre Charbit}{Université Paris Diderot, IRIF, France}{}{}{}

\author{St\'ephan Thomass\'e}{Univ Lyon, CNRS, ENS de Lyon, Université Claude Bernard Lyon 1, LIP UMR5668, France\\
Institut Universitaire de France}{}{}{}

\author{R\'emi Watrigant}{Univ Lyon, CNRS, ENS de Lyon, Université Claude Bernard Lyon 1, LIP UMR5668, France}{}{}{}
\authorrunning{\'E. Bonnet, N. Bousquet, P. Charbit, S. Thomass\'e, R. Watrigant}%mandatory. First: Use abbreviated first/middle names. Second (only in severe cases): Use first author plus 'et. al.'

\Copyright{\'E. Bonnet, N. Bousquet, P. Charbit, S. Thomass\'e, R. Watrigant}%mandatory, please use full first names. LIPIcs license is "CC-BY";  http://creativecommons.org/licenses/by/3.0/

\subjclass{Theory of computation $\rightarrow$ Fixed parameter tractability}% mandatory: Please choose ACM 2012 classifications from https://www.acm.org/publications/class-2012 or https://dl.acm.org/ccs/ccs_flat.cfm . E.g., cite as "General and reference $\rightarrow$ General literature" or \ccsdesc[100]{General and reference~General literature}. 

%mandatory

\category{}%optional, e.g. invited paper

\relatedversion{}%optional, e.g. full version hosted on arXiv, HAL, or other respository/website

\supplement{}%optional, e.g. related research data, source code, ... hosted on a repository like zenodo, figshare, GitHub, ...

\funding{\'E. B. is supported by the LABEX MILYON (ANR-10- LABX-0070) of
Université de Lyon, within the program ``Investissements d’Avenir'' (ANR-11-IDEX-0007) operated by the French National Research Agency (ANR).
N. B. and P. C. are supported by the ANR Project DISTANCIA (ANR-17-CE40-0015) operated by the French National Research Agency (ANR). }%optional, to capture a funding statement, which applies to all authors. Please enter author specific funding statements as fifth argument of the \author macro.

\acknowledgements{}%optional

\EventEditors{}
\EventNoEds{0}
\EventLongTitle{}
\EventShortTitle{}
\EventAcronym{}
\EventYear{}
\EventDate{}
\EventLocation{}
\EventLogo{}
\SeriesVolume{}
\ArticleNo{}
\nolinenumbers %uncomment to disable line numbering
\usepackage{xspace}
\usepackage[table]{xcolor}
\usepackage{tikz}
\usepackage{array}
\usepackage[shortlabels]{enumitem}
\usetikzlibrary{fit}
\usetikzlibrary{calc}
\usetikzlibrary{shapes}

\makeatletter
\newcommand\footnoteref[1]{\protected@xdef\@thefnmark{\ref{#1}}\@footnotemark}
\makeatother

\newcommand{\indepset}{\textsc{Maximum Independent Set}\xspace}
\newcommand{\shortindepset}{\textsc{MIS}\xspace}
\newcommand{\windepset}{\textsc{Max Weighted Independent Set}\xspace}

\newenvironment{fitexpmis}[1]{\textsc{#1-Iterative Expansion MIS}\xspace}{}
\newenvironment{faugramsey}[1]{\textsc{#1-Ramsey-extracted Iterative Expansion MIS}\xspace}{}

\theoremstyle{plain}
\newtheorem{conjecture}[theorem]{Conjecture}
\newtheorem{proposition}[theorem]{Proposition}

\newenvironment{claim}[1]{\medskip\par\noindent\textit{\textbf{Claim:}}\space#1}{}
\newenvironment{claimproof}[1]{\par\noindent\textit{Proof of claim:}\space#1}{\hfill $\triangleleft$\medskip}

\newcommand{\ie}{\textit{i.e.}\xspace}
\newcommand{\etal}{\textit{et al.}\xspace}
\newcommand{\wrt}{\textit{w.r.t.}\xspace}

\newcommand{\C}{\mathcal{C}}
\renewcommand{\S}{\mathcal{S}}
\newcommand{\X}{\mathcal{X}}
\newcommand{\G}{\mathcal{G}}
\newcommand{\B}{\mathcal{B}}
\newcommand{\N}{\mathbb{N}}
\newcommand{\Ramsey}{Ram}

\begin{document}

\maketitle

\begin{abstract}
In this paper, we investigate the complexity of \indepset (\shortindepset) in the class of $H$-free graphs, that is, graphs excluding a fixed graph as an induced subgraph. 
Given that the problem remains $NP$-hard for most graphs $H$, we study its fixed-parameter tractability and make progress towards a dichotomy between $FPT$ and $W[1]$-hard cases.
We first show that \shortindepset remains $W[1]$-hard in graphs forbidding simultaneously $K_{1, 4}$, any finite set of cycles of length at least $4$, and any finite set of trees with at least two branching vertices.
In particular, this answers an open question of Dabrowski \etal concerning $C_4$-free graphs.
Then we extend the polynomial algorithm of Alekseev when $H$ is a disjoint union of edges to an $FPT$ algorithm when $H$ is a disjoint union of cliques.
We also provide a framework for solving several other cases, which is a generalization of the concept of \emph{iterative expansion} accompanied by the extraction of a particular structure using Ramsey's theorem. Iterative expansion is a maximization version of the so-called \emph{iterative compression}.
We believe that our framework can be of independent interest for solving other similar graph problems.
Finally, we present positive and negative results on the existence of polynomial (Turing) kernels for several graphs $H$.
\end{abstract}

\section{Introduction}\label{sec:intro}

Given a simple graph $G$, a set of vertices $S \subseteq V(G)$ is an \emph{independent set} if the vertices of this set are all pairwise non-adjacent. Finding an independent set with maximum cardinality is a fundamental problem in algorithmic graph theory, and is known as the \shortindepset problem (\shortindepset, for short)~\cite{GJ79}. In general graphs, it is not only $NP$-hard, but also not approximable within $O(n^{1-\epsilon})$ for any $\epsilon > 0$ unless $P=NP$~\cite{Zuk07}, and $W[1]$-hard~\cite{DF13} (unless otherwise stated, $n$ always denotes the number of vertices of the input graph). 
Thus, it seems natural to study the complexity of \shortindepset in restricted graph classes. One natural way to obtain such a restricted graph class is to forbid some given pattern to appear in the input. For a fixed graph $H$, we say that a graph is \emph{$H$-free} if it does not contain $H$ as an induced subgraph.
Unfortunately, it turns out that for most graphs $H$, \shortindepset in $H$-free graphs remains $NP$-hard, as shown by a very simple reduction first observed by Alekseev:

\begin{theorem}[\cite{Ale82}]\label{thm:nph}
Let $H$ be a connected graph which is neither a path nor a subdivision of the claw. Then \shortindepset is NP-hard in $H$-free graphs.
\end{theorem}

On the positive side, the case of $P_t$-free graphs has attracted a lot of attention during the last decade. While it is still open whether there exists $t \in \N$ for which \shortindepset is $NP$-hard in $P_t$-free graphs, quite involved polynomial-time algorithms were discovered for $P_5$-free graphs~\cite{LoVaVi14}, and very recently for $P_6$-free graphs~\cite{GrKlPiPi17}. In addition, we can also mention the recent following result: \shortindepset admits a subexponential algorithm running in time $2^{O(\sqrt{t n \log n})}$ in $P_t$-free graphs for every $t \in \N$ \cite{BaLoMaPiTuLe18}.

The second open question concerns the subdivision of the claw. Let $S_{i,j,j}$ be a tree with exactly three vertices of degree one, being at distance $i$, $j$ and $k$ from the unique vertex of degree three. The complexity of \shortindepset is still open in $S_{1, 2, 2}$-free graphs and $S_{1, 1, 3}$-free graphs. In this direction, the only positive results concern some subcases: it is polynomial-time solvable in $(S_{1, 2, 2}, S_{1, 1, 3}, dart)$-free graphs \cite{Ka17}, $(S_{1, 1, 3}, banner)$-free graphs and $(S_{1, 1, 3}, bull)$-free graphs \cite{KaMa17}, where $dart$, $banner$ and $bull$ are particular graphs on five vertices.

Given the large number of graphs $H$ for which the problem remains $NP$-hard, it seems natural to investigate the existence of parameterized algorithms\footnote{For the sake of simplicity, ``\shortindepset'' will denote the optimisation, decision and parameterized version of the problem (in the latter case, the parameter is the size of the solution), the correct use being clear from the context.}, that is, determining the existence of an independent set of size $k$ in a graph with $n$ vertices in time $O(f(k)n^c)$ for some computable function $f$ and constant $c$.
A very simple case concerns $K_r$-free graphs, that is, graphs excluding a clique of size $r$. In that case, Ramsey's theorem implies that every such graph $G$ admits an independent set of size $\Omega(n^{\frac{1}{r-1}})$, where $n=|V(G)|$. In the $FPT$ vocabulary, it implies that \shortindepset in $K_r$-free graphs has a kernel with $Ok^{r-1})$ vertices.

To the best of our knowledge, the first step towards an extension of this observation within the $FPT$ framework is the work of Dabrowski \etal~\cite{Dabrowski12} (see also Dabrowski's PhD manuscript~\cite{DabrowskiThesis}) who showed, among others, that for any positive integer $r$, \windepset is FPT in $H$-free graphs when $H$ is a clique of size $r$ minus an edge.
In the same paper, they settle the parameterized complexity of \shortindepset on almost all the remaining cases of $H$-free graphs when $H$ has at most four vertices.
The conclusion is that the problem is $FPT$ on those classes, except for $H=C_4$ which is left open. We answer this question by showing that \shortindepset remains $W[1]$-hard in a subclass of $C_4$-free graphs. On the negative side, it was proved that \shortindepset remains $W[1]$-hard in $K_{1,4}$-free graphs \cite{HeMnLe14}

Finally, we can also mention the case where $H$ is the \emph{bull} graph, which is a triangle with a pending vertex attached to two different vertices. For that case, a polynomial Turing kernel was obtained~\cite{ThTrVu17} then improved~\cite{PedCrSa18}.

\subsection{Our results}

In Section~\ref{sec:W-hardness}, we present three reductions proving $W[1]$-hardness of \shortindepset in graph excluding several graphs as induced subgraphs, such as $K_{1, 4}$, any fixed cycle of length at least four, and any fixed tree with two branching vertices. We propose a definition of a graph decomposition whose aim is to capture all graphs which can be excluded using our reductions. 

In Section~\ref{sec:disjoint-union-cliques}, we extend the polynomial algorithm of Alekseev when $H$ is a disjoint union of edges to an $FPT$ algorithm when $H$ is a disjoint union of cliques.

In Section~\ref{sec:positive-two}, we present a general framework extending the technique of \emph{iterative expansion}, which itself is the maximization version of the well-known iterative compression technique. We apply this framework to provide $FPT$ algorithms when $H$ is a clique minus a complete bipartite graph, a clique minus a triangle, and when $H$ is the so-called \emph{gem} graph.

Finally, in Section~\ref{sec:kernels}, we focus on the existence of polynomial (Turing) kernels. 
We first strenghten some results of the previous section by providing polynomial (Turing) kernels in the case where $H$ is a clique minus a claw.
Then, we prove that for many $H$, \shortindepset on $H$-free graphs does not admit a polynomial kernel, unless $NP \subseteq coNP/poly$.
Our results allows to obtain the complete dichotomy polynomial/polynomial kernel (PK)/no PK but polynomial Turing kernel/$W[1]$-hard for all possible graphs on four vertices, while only five graphs on five vertices remain open for the $FPT$/$W[1]$-hard dichotomy.

\subsection{Notation}

For classical notation related to graph theory or fixed-parameter tractable algorithms, we refer the reader to the monographs~\cite{Die12} and~\cite{DF13}, respectively.
For an integer $r \ge 2$ and a graph $H$ with vertex set $V(H)=\{v_1, \dots, v_{n_H}\}$ with $n_H \le r$, we denote by $K_r \setminus H$ the graph with vertex set $\{1, \dots, r\}$ and edge set $\{ab : 1 \le a, b \le r$ such that $v_av_b \notin E(H)\}$.
For $X \subseteq V(G)$, we write $G \setminus X$ to denote $G[V(G) \setminus X]$.
For two graphs $G$ and $H$, we denote by $G \uplus H$ the \emph{disjoint union} operation, that is, the graph with vertex set $V(G) \cup V(H)$ and edge set $E(G) \cup E(H)$. We denote by $G + H$ the \emph{join} operation of $G$ and $H$, that is, the graph with vertex set $V(G) \cup V(H)$ and edge set $E(G) \cup E(H) \cup \{uv: u \in V(G), v \in V(H)\}$.
For two integers $r, k$, we denote by $\Ramsey(r, k)$ the Ramsey number of $r$ and $k$, \ie the minimum order of a graph to contain either a clique of size $r$ or an independent set of size $k$. We write for short $\Ramsey(k)=\Ramsey(k,k)$. Finally, for $\ell, k > 0$, we denote by $\Ramsey_{\ell}(k)$ the minimum order of a complete graph whose edges are colored with $\ell$ colors to contain a monochromatic clique of size $k$.

\section{$W[1]$-hardness}\label{sec:W-hardness}

\subsection{Main reduction}
We have the following:

\begin{theorem}\label{thm:Whard}
For any $p_1 \ge 4$ and $p_2 \ge 1$, \shortindepset remains $W[1]$-hard in graphs excluding simultaneously the following graphs as induced subgraphs:
\begin{itemize}
	\item $K_{1, 4}$
	\item $C_4$, $\dots$, $C_{p_1}$
	\item any tree $T$ with two branching vertices\footnote{A branching vertex in a tree is a vertex of degree at least $3$.} at distance at most $p_2$.
\end{itemize}
\end{theorem}
\begin{proof}
Let $p = \max\{p_1, p_2\}$. We reduce from \textsc{Grid Tiling}, where the input is composed of $k^2$ sets $S_{i,j} \subseteq [m] \times [m]$ ($0 \le i, j \le k-1$), called \emph{tiles}, each composed of $n$ elements. The objective of \textsc{Grid Tiling} is to find an element $s^*_{i,j} \in S_{i, j}$ for each $0 \le i, j \le k-1$, such that $s^*_{i,j}$ agrees in the first coordinate with $s^*_{i, j+1}$, and agrees in the second coordinate with $s^*_{i+1, j}$, for every $0 \le i, j \le k-1$ (incrementations of $i$ and $j$ are done modulo $k$). In such case, we say that $\{s^*_{i,j}$, $0 \le i,j \le k-1\}$ is a \emph{feasible solution} of the instance.
It is known that \textsc{Grid Tiling} is $W[1]$-hard parameterized by $k$ \cite{CyFoKoLoMaPiPiSa15}. 

Before describing formally the reduction, let us give some definitions and ideas. Given $s=(a,b)$ and $s'=(a',b')$, we say that $s$ is \emph{row-compatible} (resp. \emph{column-compatible}) with $s'$ if $a \ge a'$ (resp. $b \ge b'$)\footnote{Notice that the row-compatibility (resp. column-compatibility) relation is not symmetrical.}. Observe that a solution $\{s^*_{i,j}$, $0 \le i,j \le k-1\}$ is feasible if and only if $s^*_{i,j}$ is row-compatible with $s^*_{i, j+1}$ and column-compatible with $s^*_{i+1, j}$ for every $0 \le i,j \le k-1$ (incrementations of $i$ and $j$ are done modulo $k$). Informally, the main idea of the reduction is that, when representing a tile by a clique, the row-compatibility (resp. column-compatibility) relation (as well at its complement) forms a $C_4$-free graph when considering two consecutive tiles, and a claw-free graph when considering three consecutive tiles. The main difficulty is to forbid the desired graphs to appear in the ``branchings'' of tiles. We now describe the reduction. 
% We will then "subdivide" the adjacencies between tiles in order to increase the distance between two claws, and the size of a minimum cycle of length at least $4$ (using half graphs for the "propagation" of the choice of a vertex along the subdivided paths).

\begin{center}
\begin{figure}
\includegraphics[scale=0.77]{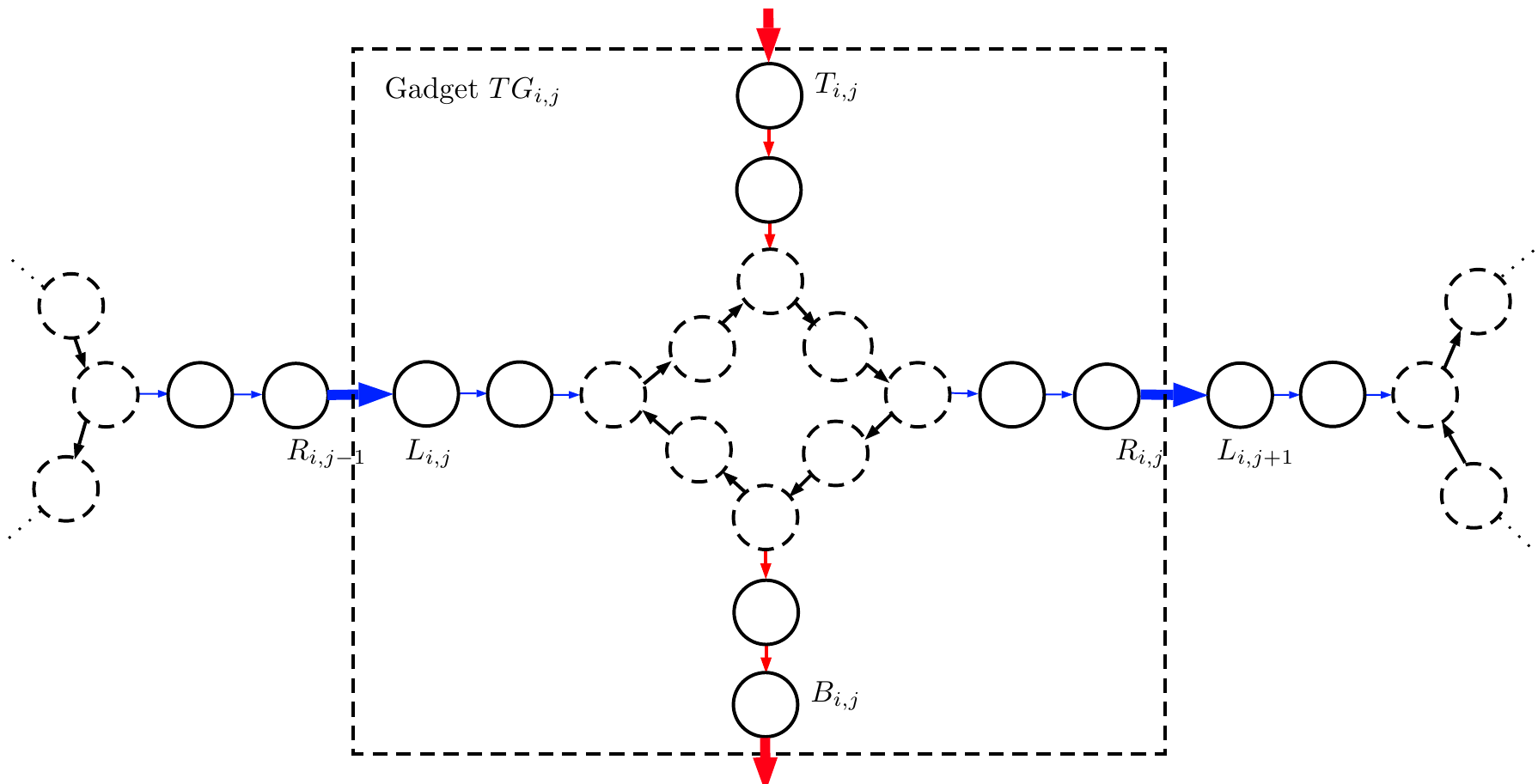}
\caption{Gadget $TG^{i,j}$ representing a tile and its adjacencies with $TG_{i,j-1}$ and $TG_{i, j+1}$, for $p=1$. Each circle is a clique on $n$ vertices (dashed cliques are the cycle cliques). Black, blue and red arrows represent respectively type $T_h$, $T_r$ and $T_c$ edges (bold arrows are between two gadgets). Figures~\ref{fig:halfgraph} and \ref{fig:tile} represent some adjacencies in more details.}
\label{fig:gadget}
\end{figure}
\end{center}

\begin{center}
\begin{figure}
\begin{subfigure}[b]{0.45\textwidth}
\includegraphics[scale=0.2]{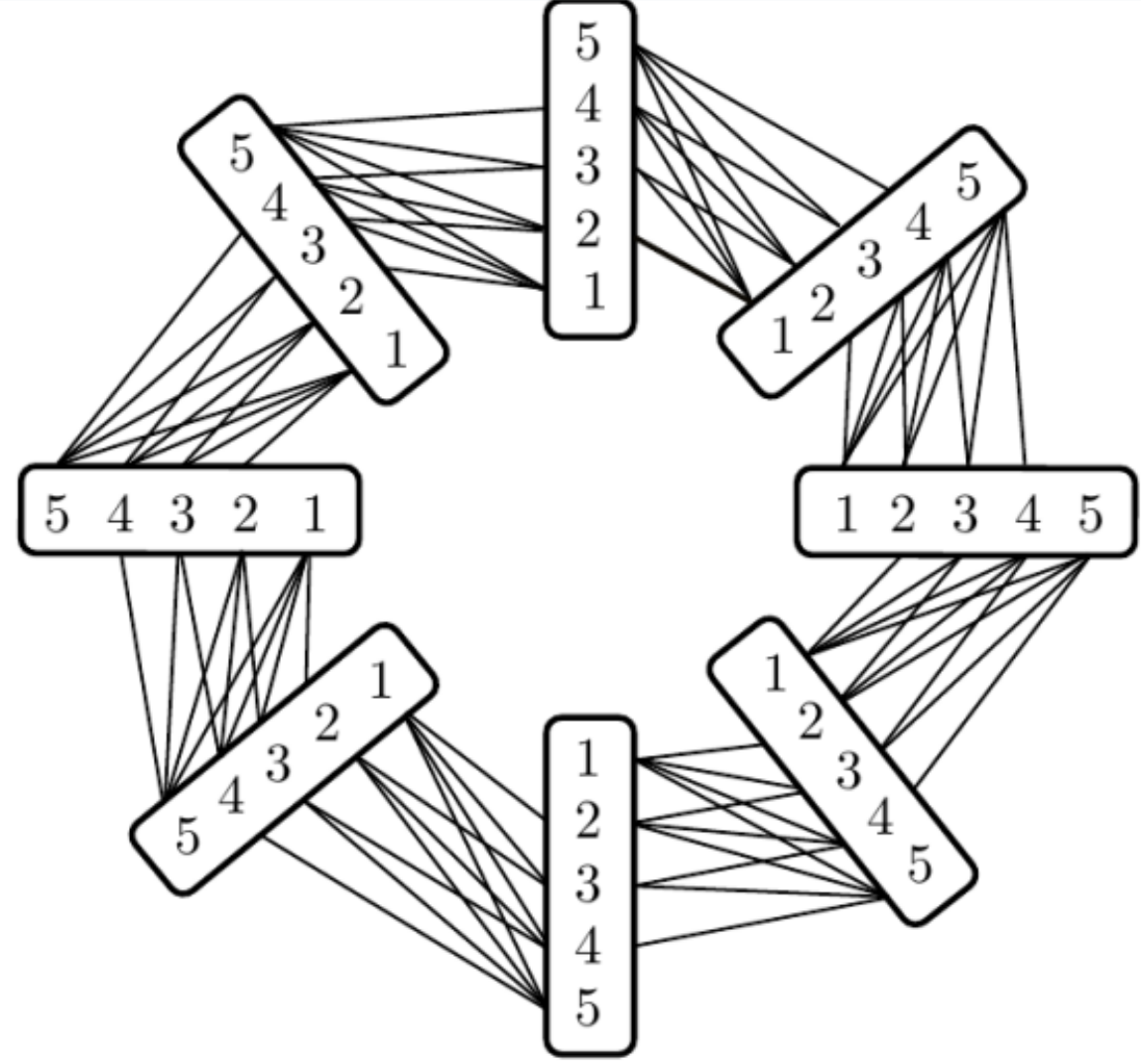}
\caption{Adjacencies between cycle cliques (represented by dashed circles in Figure~\ref{fig:gadget}).}
\label{fig:halfgraph}
\end{subfigure}
\begin{subfigure}[b]{0.55\textwidth}
\includegraphics[scale=1.1]{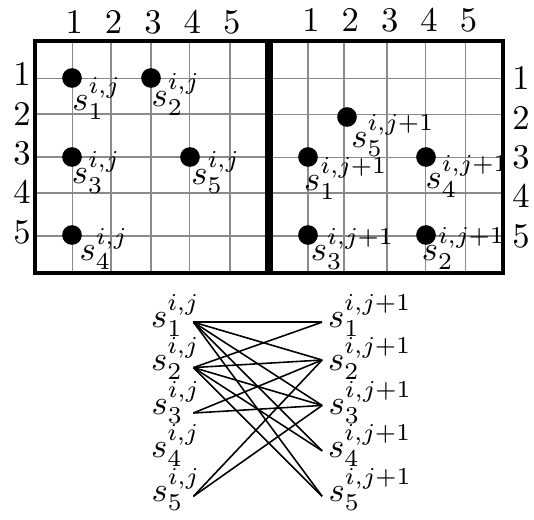}
\caption{Two consecutive tiles and the representation of their adjacencies (representing type $T_r$ adjacencies).}
\label{fig:tile}
\end{subfigure}
\caption{Some example of adjacencies within the first reduction.}
\end{figure}
\end{center}

\subsubsection{Tile gadget.}
For every tile $S_{i, j} = \{s^{i,j}_1, \dots, s^{i, j}_n\}$, we construct a \emph{tile gadget} $TG_{i, j}$, depicted in Figure~\ref{fig:gadget}. 
Notice that this gadget shares some ideas with the $W[1]$-hardness of the problem in $K_{1,4}$-free graphs by Hermelin \etal \cite{HeMnLe14}.
To define this gadget, we first describe an oriented graph with three types of arcs (type $T_h$, $T_r$ and $T_c$, which respectively stands for \emph{half graph}, \emph{row} and \emph{column}, this meaning will become clearer later), and then explain how to represent the vertices and arcs of this graph to get the concrete gadget. 
Consider first a directed cycle on $4p+4$ vertices $c_1$, $\dots$, $c_{4p+4}$ with arcs of type $T_h$. 
Then consider four oriented paths on $p+1$ vertices: $P_1$, $P_2$, $P_3$ and $P_4$. $P_1$ and $P_3$ are composed of arcs of type $T_c$, while $P_2$ and $P_4$ are composed of arcs of type $T_r$.
Put an arc of type $T_c$ between:
\begin{itemize}
	\item the last vertex of $P_1$ and $c_1$,
	\item $c_{2p+3}$ and the first vertex of $P_3$,
\end{itemize}
and an arc of type $T_r$ between:
\begin{itemize}
	\item $c_{p+2}$ and the first vertex of $P_2$,
	\item the last vertex of $P_4$ and $c_{3p+4}$.
\end{itemize}
Now, replace every vertex of this oriented graph by a clique on $n$ vertices, and fix an arbitrary ordering on the vertices of each clique.
For each arc of type $T_h$ between $c$ and $c'$, add a half graph\footnote{Notice that our definition of half graph slighly differs from the usual one, in the sense that we do not put edges relying two vertices of the same index. Hence, our construction can actually be seen as the complement of a half graph (which is consistent with the fact that usually, both parts of a half graph are independent sets, while they are cliques in our gadgets).} between the corresponding cliques: connect the $a^{th}$ vertex of the clique representing $c$ with the $b^{th}$ vertex of the clique representing $c'$ iff $a > b$. 
%For convenience, we still view the half graph from a clique $C$ to a clique $C'$ as an arc of type $T_h$ from $C$ to $C'$.
For every arc of type $T_r$ from a vertex $c$ to a vertex $c'$, connect the $a^{th}$ vertex of the clique representing $c$ with the $b^{th}$ vertex of the clique representing $c'$ iff $s_a^{i,j}$ is \emph{not} row-compatible with $s_b^{i,j}$.
Similarly, for every arc of type $T_c$ from a vertex $c$ to a vertex $c'$, connect the $a^{th}$ vertex of the clique representing $C$ with the $b^{th}$ vertex of the clique representing $c'$ iff $s_a^{i,j}$ is \emph{not} column-compatible with $s_b^{i,j}$.
The cliques corresponding to vertices of this gadget are called the \emph{main cliques} of $TG_{i,j}$, and the cliques corresponding to the central cycle on $4p+4$ vertices are called the \emph{cycle cliques}.
The main cliques which are not cycle cliques are called \emph{path cliques}.
The cycle cliques adjacent to one path clique are called \emph{branching cliques}.
Finally, the clique corresponding to the vertex of degree one in the path attached to $c_1$ (resp. $c_{p+2}$, $c_{2p+3}$, $c_{3p+4}$) is called the \emph{top} (resp. \emph{right}, \emph{bottom}, \emph{left}) clique of $TG_{i,j}$,  denoted by $T_{i,j}$ (resp. $R_{i,j}$, $B_{i,j}$, $L_{i,j}$). 
Let $T_{i,j}=\{t^{i,j}_1, \dots, t^{i,j}_n\}$, $R_{i,j}=\{r^{i,j}_1, \dots, r^{i,j}_n\}$, $B_{i,j}=\{b^{i,j}_1, \dots, b^{i,j}_n\}$, and $L_{i,j}=\{\ell^{i,j}_1, \dots, \ell^{i,j}_n\}$. For the sake of readability, we might omit the superscripts $i,j$ when it is clear from the context.

\begin{lemma}\label{lemma:propagationgadget}
Let $K$ be an independent set of size $8(p+1)$ in $TG_{i,j}$. Then:
\begin{enumerate}[(a)]
	\item $K$ intersects all the cycle cliques on the same index $x$;
	\item if $K \cap T_{i,j} = \{t_{x_t}\}$, $K \cap R_{i,j} = \{r_{x_r}\}$, $K \cap B_{i,j} = \{b_{x_b}\}$, and $K \cap L_{i,j} = \{\ell_{x_{\ell}}\}$. Then:
\begin{itemize}
	\item $s^{i,j}_{x_{\ell}}$ is row-compatible with $s^{i,j}_x$ which is row-compatible with $s^{i,j}_{x_r}$, and
	\item $s^{i,j}_{x_t}$ is column-compatible with $s^{i,j}_x$ which is column-compatible with $s^{i,j}_{x_b}$.
\end{itemize}
\end{enumerate} 
\end{lemma}
\begin{proof}
Observe that the vertices of $TG_{i,j}$  can be partitionned into $8(p+1)$ cliques (the main cliques), hence an independent set of size $8(p+1)$ intersects each main clique on exactly one vertex. 
Let $C_1$, $C_2$ and $C_3$ be three consecutive cycle cliques, and suppose $K$ intersects $C_1$ (resp. $C_2$, $C_3$) on the $x_1^{th}$ (resp. $x_2^{th}$, $x_3^{th}$) index. By definition of the gadget, it implies $x_1 \le x_2 \le x_3$. By applying the same argument from $C_3$ along the cycle, we obtain $x_3 \le x_1$, which proves (a).
The proof of (b) directly comes from the definition of the adjacencies between cliques of type $T_r$ and $T_c$, and from the fact that $K$ intersects all cycle cliques on the same index.
\end{proof}

\subsubsection{Attaching gadgets together.} For $i,j \in \{0, \dots, k-1\}$, we connect the right clique of $TG_{i,j}$ with the left clique of $TG_{i, j+1}$ in a ``type $T_r$ spirit'': for every $x, y \in [n]$, connect $r_x^{i,j} \in R_{i,j}$ with $\ell_y^{i,j+1} \in L_{i,j+1}$ iff $s_{x}^{i,j}$ is \emph{not} row-compatible with $s_{y}^{i,j+1}$. Similarly, we connect the bottom clique of $TG_{i,j}$ with the top clique of $TG_{i+1, j}$ in a ``type $T_c$ spirit'': for every $x,y \in [n]$, connect $b_x^{i,j} \in B_{i,j}$ with $t_y^{i+1,j} \in T_{i+1,j}$ iff $s_{x}^{i,j}$ is \emph{not} column-compatible with $s_{y}^{i+1,j}$ (all incrementations of $i$ and $j$ are done modulo $k$).
 This terminates the construction of the graph $G$. 
 
 \subsubsection{Equivalence of solutions.}
 We now prove that the input instance of \textsc{Grid Tiling} is positive if and only if $G$ has an independent set of size $k'=8(p+1)k^2$.
First observe that $G$ has $k^2$ tile gadgets, each composed of $8(p+1)$ main cliques, hence any independent set of size $k'$ intersects each main clique on exactly one vertex.
 By Lemma~\ref{lemma:propagationgadget}, for all $i,j \in \{0, \dots, k-1\}$, $K$ intersects the cycle cliques of $TG_{i,j}$ on the same index $x_{i,j}$. 
 Moreover, if $K \cap R_{i,j} = \{r_x^{i,j}\}$ and $K \cap L_{i,j+1} = \{\ell_{x'}^{i,j+1}\}$, then, by construction of $G$, $s^{i,j}_x$ is row-compatible with $s_{x'}^{i,j+1}$. Similarly, if $K \cap B_{i,j} = \{b_x^{i,j}\}$ and $K \cap T_{i+1,j} = \{t_{x'}^{i+1,j}\}$, then, by construction of $G$, $s^{i,j}_x$ is column-compatible with $s_{x'}^{i+1,j}$. By Lemma~\ref{lemma:propagationgadget}, it implies that $s^{i,j}_{x_{i,j}}$ is row-compatible with $s^{i,j+1}_{x_{i,j+1}}$ and column-compatible with $s^{i+1,j}_{x_{i+1,j}}$ (incrementations of $i$ and $j$ are done modulo $k$), thus $\{x^{i,j}_{x_{i,j}} : 0 \le i,j \le k-1 \}$ is a feasible solution. 
Using similar ideas, one can prove that a feasible solution of the grid tiling instance implies an independent set of size $k'$ in $G$.

\subsubsection{Structure of the obtained graph.}
Let us now prove that $G$ does not contain the graphs mentionned in the statement as an induced subgraph:
\begin{enumerate}[(i)]
	\item $K_{1,4}$: we first prove that for every $0 \le i,j \le k-1$, the graph induced by the cycle cliques of $TG_{i,j}$ is claw-free. For the sake of contradiction, suppose that there exist three consecutive cycle cliques $A$, $B$ and $C$ containing a claw. W.l.o.g. we may assume that $b_x \in B$ is the center of the claw, and $a_{\alpha} \in A$, $b_{\beta} \in B$ and $c_{\gamma} \in C$ are the three endpoints. By construction of the gadgets (there is a half graph between $A$ and $B$ and between $B$ and $C$), we must have $\alpha < x < \gamma$. Now, observe that if $x < \beta$ then $a_{\alpha}$ must be adjacent to $b_{\beta}$, and if $\beta < x$, then $b_{\beta}$ must be adjacent to $c_{\gamma}$, but both case are impossible since $\{a_{\alpha}, b_{\beta}, c_{\gamma}\}$ is supposed to be an independent set. Similarly, we can prove that the graph induced by each path of size $2(p+1)$ linking two consecutive gadgets is claw-free. Hence, the only way for $K_{1,4}$ to appear in $G$ would be that the center appears in the cycle clique attached to a path, for instance in the clique represented by the vertex $c_1$ in the cycle. However, it can easily be seen that in this case, a claw must lie either in the graph induced by the cycle cliques of the gadget, or in the path linking $TG_{i,j}$ with $TG_{i-1, j}$, which is impossible.
	
	\item $C_4$, $\dots$, $C_{p_1}$. The main argument is that the graph induced by any two main cliques does not contain any of these cycles. Then, we show that such a cycle cannot lie entirely in the cycle cliques of a single gadget $TG_{i,j}$. Indeed, if this cycle uses at most one vertex per main clique, then it must be of length at least $4p+4$. If it intersects a clique $C$ on two vertices, then either it also intersect all the cycle cliques of the gadget, in which case it is of length $4p+5$, or it intersects an adjacent clique of $C$ on two vertices, in which case these two cliques induce a $C_4$, which is impossible. Similarly, such a cycle cannot lie entirely in a path between the main cliques of two gadgets. Finally, the main cliques of two gadgets are at distance $2(p+1)$, hence such a cycle cannot intersect the main cliques of two gadgets.
	
	\item any tree $T$ with two branching vertices at distance at most $p_2$. Using the same argument as for the $K_{1,4}$ case, observe that the claws contained in $G$ can only appear in the cycle cliques where the paths are attached. However, observe that these cliques are at distance $2(p+1) > p_2$, thus, such a tree $T$ cannot appear in $G$.
\end{enumerate}
\end{proof}

As a direct consequence of Theorem~\ref{thm:Whard}, we get the following by setting $p_1=p_2=|V(H)|+1$:
\begin{corollary}\label{cor:hardCases1}
  If $H$ is not chordal, or contains as an induced subgraph a $K_{1,4}$ or a tree with two branching vertices, then \shortindepset in $H$-free graphs is $W[1]$-hard. 
\end{corollary}

\subsection{Capturing Hard Graphs}\label{sec:harddecompositions}

We introduce two variants of the hardness construction of Theorem~\ref{thm:Whard}, which we refer to as the \emph{first construction}.
The \emph{second construction} is obtained by replacing each interaction between two main cliques by an anti-matching, except the one interaction in the middle of the path cliques which remains a half-graph (see Figure~\ref{fig:hardnessConstructionVariants}, middle).
In an anti-matching, the same elements in the two adjacent cliques define the only non-edges.
The correctness of this new reduction is simpler since the propagation of a choice is now straightforward.
Observe however that the graph $C_4$ appears in this new construction.
For the \emph{third construction}, we start from the second construction and just add an anti-matching between two neighbors of each branching clique among the cycle cliques (see Figure~\ref{fig:hardnessConstructionVariants}, right).
This anti-matching only constrains more the instance but does not destroy the intended solutions; hence the correctness.

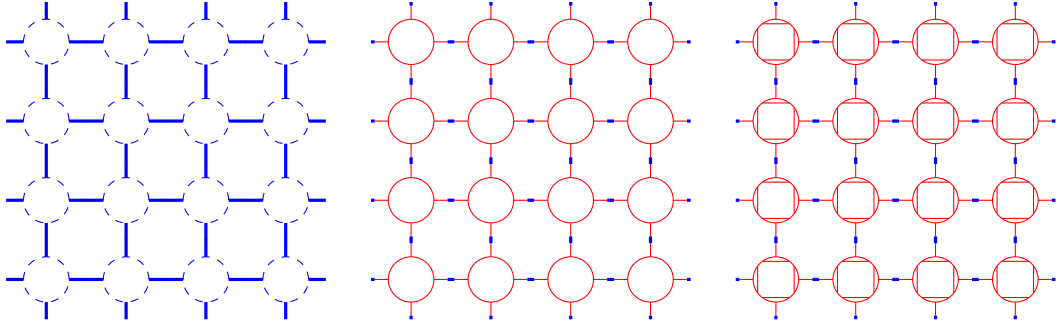
\begin{figure}
  \centering
  \begin{tikzpicture}[scale=0.6]
    \def\k{4}
    \def\r{0.5}
    \def\s{1.75}
    \def\t{0.3}
    \def\u{0.38}
    \foreach \i in {1,...,\k}{
      \foreach \j in {1,...,\k}{
        \draw[dashed, blue] (\s * \i, \s * \j) circle (\r) ;
        \draw[very thick, blue] (\s * \i - \r, \s * \j) -- (\s * \i - \r - \u, \s * \j) ;
        \draw[very thick,blue] (\s * \i + \r, \s * \j) -- (\s * \i + \r + \u, \s * \j) ;
        \draw[very thick,blue] (\s * \i, \s * \j - \r - \u) -- (\s * \i, \s * \j - \r) ;
        \draw[very thick,blue] (\s * \i, \s * \j + \r + \u) -- (\s * \i, \s * \j + \r) ;
      }
    }

    \begin{scope}[xshift=8cm]
    \foreach \i in {1,...,\k}{
      \foreach \j in {1,...,\k}{
        \draw[red] (\s * \i, \s * \j) circle (\r) ;
        \draw[red] (\s * \i - \r, \s * \j) -- (\s * \i - \r - \t, \s * \j) ;
        \draw[red] (\s * \i + \r, \s * \j) -- (\s * \i + \r + \t, \s * \j) ;
        \draw[red] (\s * \i, \s * \j - \r - \t) -- (\s * \i, \s * \j - \r) ;
        \draw[red] (\s * \i, \s * \j + \r + \t) -- (\s * \i, \s * \j + \r) ;
        \draw[very thick, blue] (\s * \i - \r - \u, \s * \j) -- (\s * \i - \r - \t, \s * \j) ;
        \draw[very thick, blue] (\s * \i + \r + \u, \s * \j) -- (\s * \i + \r + \t, \s * \j) ;
        \draw[very thick,blue] (\s * \i, \s * \j - \r - \t) -- (\s * \i, \s * \j - \r - \u) ;
        \draw[very thick,blue] (\s * \i, \s * \j + \r + \t) -- (\s * \i, \s * \j + \r + \u) ;
      }
    }
    \end{scope}

    \def\v{0.8}
    \def\w{0.3}
    \begin{scope}[xshift=16cm]
    \foreach \i in {1,...,\k}{
      \foreach \j in {1,...,\k}{
        \draw[red] (\s * \i, \s * \j) circle (\r) ;
        \draw[red] (\s * \i - \v * \r, \s * \j - \w) -- (\s * \i - \v * \r, \s * \j + \w) ;
        \draw[red] (\s * \i + \v * \r, \s * \j - \w) -- (\s * \i + \v * \r, \s * \j + \w) ;
        \draw[red] (\s * \i - \w, \s * \j - \v * \r) -- (\s * \i + \w, \s * \j - \v * \r) ;
        \draw[red] (\s * \i - \w, \s * \j + \v * \r) -- (\s * \i + \w, \s * \j + \v * \r) ;
        \draw[red] (\s * \i - \r, \s * \j) -- (\s * \i - \r - \t, \s * \j) ;
        \draw[red] (\s * \i + \r, \s * \j) -- (\s * \i + \r + \t, \s * \j) ;
        \draw[red] (\s * \i, \s * \j - \r - \t) -- (\s * \i, \s * \j - \r) ;
        \draw[red] (\s * \i, \s * \j + \r + \t) -- (\s * \i, \s * \j + \r) ;
        \draw[very thick, blue] (\s * \i - \r - \u, \s * \j) -- (\s * \i - \r - \t, \s * \j) ;
        \draw[very thick, blue] (\s * \i + \r + \u, \s * \j) -- (\s * \i + \r + \t, \s * \j) ;
        \draw[very thick,blue] (\s * \i, \s * \j - \r - \t) -- (\s * \i, \s * \j - \r - \u) ;
        \draw[very thick,blue] (\s * \i, \s * \j + \r + \t) -- (\s * \i, \s * \j + \r + \u) ;
      }
    }
    \end{scope}
  \end{tikzpicture}
  \caption{A symbolic representation of the hardness constructions. To the left, only half-graphs (blue) are used between the cliques, as in the proof of Theorem~\ref{thm:Whard}.
    %In particular, this prevents the $C_4$ between the cliques.
  In the middle and to the right, the half-graphs (blue) are only used once in the middle of each path of cliques, and the rest of the interactions between the cliques are anti-matchings (red). The third construction (right) is a slight variation of the second (middle) where for each branching clique, we link by an anti-matching its two neighbors among the cycle cliques.}
  \label{fig:hardnessConstructionVariants}
\end{figure}

To describe those connected graphs $H$ which escape the disjunction of Theorem~\ref{thm:Whard} (for which there is still a hope that \shortindepset is FPT), we define a decomposition into cliques, similar yet different from clique graphs or tree decompositions of chordal graphs (a.k.a $k$-trees).

\begin{definition}\label{def:cliqueDec}
Let $T$ be a graph on $\ell$ vertices $t_1, \ldots, t_\ell$.
We say that \emph{$T$ is a clique decomposition of $H$} if there is a partition of $V(H)$ into $(C_1,C_2,\ldots,C_\ell)$ such that:
\begin{itemize}
 \item for each $i \in [\ell]$, $H[C_i]$ is a clique, and
 \item for each pair $i \neq j \in [\ell]$, if $H[C_i \cup C_j]$ is connected, then $t_it_j \in E(T)$. 
\end{itemize}
\end{definition}
Observe that, in the above definition, we do not require $T$ to be a tree.
Two cliques $C_i$ and $C_j$ are said \emph{adjacent} if $H[C_i \cup C_j]$ is connected.
We also write \emph{a clique decomposition on $T$ (of $H$)} to denote the choice of an actual partition $(C_1,C_2,\ldots,C_\ell)$.  

Let $\mathcal T_1$ be the class of trees with at most one branching vertex.
Equivalently, $\mathcal T_1$ consists of the paths and the subdivisions of the claw.

\begin{proposition}\label{prop:cliqueDecWHard}
 For a fixed connected graph $H$, if no tree in $\mathcal T_1$ is a clique decomposition of $H$, then \shortindepset in $H$-free graphs is $W[1]$-hard. 
\end{proposition}
\begin{proof}
  This is immediate from the proof of Theorem~\ref{thm:Whard} since $H$ cannot appear in the first construction. 
\end{proof}

At this point, we can focus on connected graphs $H$ admitting a tree $T \in \mathcal T_1$ as a clique decomposition.
The reciprocal of Proposition~\ref{prop:cliqueDecWHard} cannot be true since a simple edge is a clique decomposition of $C_4$.  
The next definition further restricts the interaction between two adjacent cliques.

\begin{definition}\label{def:StrongCliqueDec}
Let $T$ be a graph on $\ell$ vertices $t_1, \ldots, t_\ell$.
We say that \emph{$T$ is a strong clique decomposition of $H$} if there is a partition of $V(H)$ into $(C_1,C_2,\ldots,C_\ell)$ such that:
\begin{itemize}
 \item for each $i \in [\ell]$, $H[C_i]$ is a clique, and
 \item for each pair $i \neq j \in [\ell]$, $H[C_i \cup C_j]$ is a clique iff $t_it_j \in E(T)$. 
\end{itemize}
\end{definition}

An equivalent way to phrase this definition is that $H$ can be obtained from $T$ by \emph{adding false twins}.
Adding a false twin $v'$ to a graph consists in duplicating one of its vertex $v$ (i.e., $v$ and $v'$ have the same neighbors) and then adding an edge between $v$ and $v'$.

We define \emph{almost strong clique decompositions} which informally are strong clique decompositions where at most one edge can be missing in the interaction between two adjacent cliques.
\begin{definition}\label{def:AlmostStrongCliqueDec}
Let $T$ be a graph on $\ell$ vertices $t_1, \ldots, t_\ell$.
We say that \emph{$T$ is an almost strong clique decomposition of $H$} if there is a partition of $V(H)$ into $(C_1,C_2,\ldots,C_\ell)$ such that:
\begin{itemize}
 \item for each $i \in [\ell]$, $H[C_i]$ is a clique, and
 \item for each pair $i \neq j \in [\ell]$, [$H[C_i \cup C_j]$ is a clique or $H[C_i \cup C_j]$ is a clique of size at least 3 minus an edge] iff $t_it_j \in E(T)$. 
\end{itemize}
\end{definition}

Finally, a \emph{nearly strong clique decomposition} is slightly weaker than an almost strong clique decomposition: at most one interaction between two adjacent cliques can induce a $C_4$-free graph. 

Let $\mathcal P$ be the set of all the paths.
Notice that $\mathcal T_1 \setminus \mathcal P$ is the set of all the subdivisions of the claw.
\begin{theorem}\label{thm:closerToDichotomy}
  Let $H$ be a fixed connected graph.
  If no $P \in \mathcal P$ is a nearly strong clique decomposition of $H$  and no $T \in \mathcal T_1 \setminus \mathcal P$ is an almost strong clique decomposition of $H$, then \shortindepset in $H$-free graphs is $W[1]$-hard.
\end{theorem}
\begin{proof}
  The idea is to mainly use the second construction and the fact that \shortindepset in $C_4$-free graphs is $W[1]$-hard (due to the first construction).
  For every fixed graph $H$ which cannot be an induced subgraph in the second construction, \shortindepset is $W[1]$-hard.
  To appear in this construction, the graph $H$ should have
  \begin{itemize}
  \item either a clique decomposition on a subdivision of the claw, such that the interaction between two adjacent cliques is the complement of a (non necessarily perfect) matching, or
  \item a clique decomposition on a path, such that the interaction between two adjacent cliques is the complement of a matching, except for at most one interaction which can be a $C_4$-free graph.
  \end{itemize}
  We now just observe that in both cases if, among the interactions between adjacent cliques, one complement of matching has at least two non-edges, then $H$ contains an induced $C_4$.
  Hence the two items can be equivalently replaced by the existence of an almost strong clique decomposition on a subdivision of the claw, and a nearly strong clique decomposition on a path, respectively. 
\end{proof}
Theorem~\ref{thm:closerToDichotomy} narrows down the connected open cases to graphs $H$ which have a nearly strong clique decomposition on a path or an almost strong clique decomposition on a subdivision of the claw.

In the strong clique decomposition, the interaction between two adjacent cliques is very simple: their union is a clique.
Therefore, it might be tempting to conjecture that if $H$ admits $T \in \mathcal T_1$ as a strong clique decomposition, then \shortindepset in $H$-free graphs is FPT.
Indeed, those graphs $H$ appear everywhere in both the first and the second $W[1]$-hardness constructions.
Nevertheless, we will see that this conjecture is false: even if $H$ has a strong clique decomposition $T \in \mathcal T_1$, it can be that \shortindepset is $W[1]$-hard.
The simplest tree of $\mathcal T_1 \setminus \mathcal P$ is the claw.
We denote by $T_{i,j,k}$ the graph obtained by adding a universal vertex to the disjoint union of three cliques $K_i \uplus K_j \uplus K_k$.
The claw is a strong clique decomposition of $T_{i,j,k}$ (for every natural numbers $i, j, k$).

\begin{theorem}\label{thm:T122}
  \shortindepset in $T_{1,2,2}$-free graphs is $W[1]$-hard.
\end{theorem}
\begin{proof}
  We show that $T_{1,2,2}$ does not appear in the third construction (Figure~\ref{fig:hardnessConstructionVariants}, right).
  We claim that, in this construction, the graph $T_{1,1,2}$, sometimes called cricket, can only appear in the two ways depicted on Figure~\ref{fig:cricket} (up to symmetry).

  \begin{claim}
    The triangle of the cricket cannot appear within the same main clique.
  \end{claim}
  \begin{proof}
    Otherwise the two leaves (\ie, vertices of degree $1$) of the cricket are in two distinct adjacent cliques.
    But at least one of those adjacent cliques is linked to the main clique of the triangle by an anti-matching.
    This is a contradiction to the corresponding leaf having two non-neighbors in the main clique of the triangle.
  \end{proof}
  We first study how the cricket can appear in a path of cliques.
  Let $C$ be the main clique containing the universal vertex of the cricket.
  This vertex is adjacent to three disjoint cliques $K_1 \uplus K_1 \uplus K_2$.
  Due to the previous claim, the only way to distribute them is to put $K_1$ in the previous main clique, $K_1$ in the same main clique $C$, and $K_2$ in the next main clique.
  This is only possible if the interaction between $C$ and the next main clique is a half-graph.
  In particular, this implies that the interaction between the previous main clique and $C$ is an anti-matching.
  This situation corresponds to the left of Figure~\ref{fig:cricket}.

  This also implies that the cricket cannot appear in a path of cliques without a half-graph interaction (anti-matchings only).
  We now turn our attention to the vicinity of a triangle of main cliques, which is proper to the third construction.
  By our previous remarks, we know that the universal vertex of the cricket has to be either alone in a main clique (by symmetry, it does not matter which one) of the triangle, or with exactly one of its neighbors of degree $2$.
  Now, the only way to place $K_1 \uplus K_1 \uplus K_2$ is to put the two $K_1$ in the two other main cliques of the triangle, and the $K_2$ (or the single vertex rest of it) in the remaining adjacent main clique.
  Indeed, if the $K_2$ is in a main clique of the triangle, the $K_1$ in the third main clique of the triangle would have two non-edges towards to $K_2$.
  This is not possible with an anti-matching interaction.
  Therefore, the only option corresponds to the right of Figure~\ref{fig:cricket}.

  \begin{figure}
    \centering
    \begin{tikzpicture}
      \def\v{0.6}
      \def\h{1.2}

      %on the path
      \node (d1) at (-\h,0) {} ;
      \node (d2) at (-\h,\v) {} ;
      \node[draw, ellipse, fit=(d1) (d2)] (e1) {} ;
      \node (d3) at (0,0) {} ;
      \node[draw,circle] (a) at (0,\v) {} ;
      \node[draw, ellipse, fit=(d3) (a)] (e2) {} ;
      \node[draw,circle] (b) at (\h,\v) {} ;
      \node[draw,circle] (c) at (\h,0) {} ;
      \node[draw, ellipse, fit=(b) (c)] (e3) {} ;
      \node[draw,circle] (d) at (2 * \h,\v) {} ;
      \node[draw,circle] (e) at (2 * \h,0) {} ;
      \node[draw, ellipse, fit=(d) (e)] (e4) {} ;
      \node (d4) at (3 * \h,0) {} ;
      \node (d5) at (3 * \h,\v) {} ;
      \node[draw, ellipse, fit=(d4) (d5)] (e5) {} ;
      \draw[very thick] (c) -- (e) -- (d) -- (c) -- (b) ;
      \draw[very thick] (a) -- (c) ;
      \draw[dashed] (a) -- (b) ;
      \draw[very thick, red] (e1.east) -- (e2.west) ;
      \draw[very thick, red] (e2.east) -- (e3.west) ;
      \draw[very thick, blue] (e3.east) -- (e4.west) ;
      \draw[very thick, red] (e4.east) -- (e5.west) ;

      %on the triangle
      \foreach \i/\j/\k in {5.8/0/0,7/0/1,8.2/0/2,9/1/3,9/-1/4,10.2/1.4/5,10.2/-1.4/6,11.4/1.5/7,11.4/-1.5/8}{
        \node (d1\k) at (\i,\j) {} ;
        \node (d2\k) at (\i,\j+\v) {} ;
        \node[draw, ellipse, fit=(d1\k) (d2\k)] (e\k) {} ;
      }
      \draw[very thick, red] (e0) -- (e1) -- (e2) -- (e3) -- (e5) -- (e7) ;
      \draw[very thick, red] (e2) -- (e4) -- (e6) -- (e8) ;
      \draw[very thick, red] (e3) -- (e4) ;

      \node[draw,circle] (A) at (8.2,0) {} ;
      \node[draw,circle] (B) at (9,1) {} ;
      \node[draw,circle] (C) at (8.87,-0.6) {} ;
      \node[draw,circle] (D) at (10.2,-0.8) {} ;
      \node[draw,circle] (E) at (10.2,-1.4) {} ;
      \draw[very thick] (C) -- (E) -- (D) -- (C) -- (B) ;
      \draw[very thick] (A) -- (C) ;
      \draw[dashed] (A) -- (B) ;
    \end{tikzpicture}
    \caption{The two ways the cricket appears in the third construction. The red edges between two adjacent cliques symbolize an anti-matching, whereas the blue edge symbolizes a $C_4$-free graph.
     In the left hand-side, one neighbor of the universal vertex with degree 2 could alternatively be in the same clique as the universal vertex.}
    \label{fig:cricket}
  \end{figure}
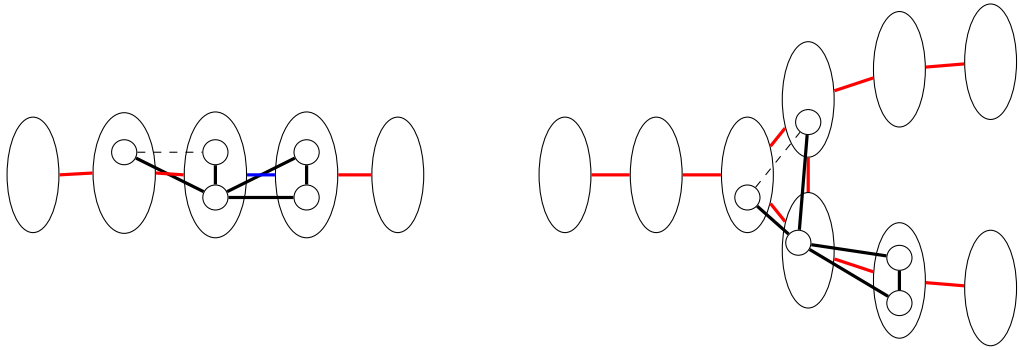
  To obtain a $T_{1,2,2}$, one needs to find a false twin to one of the leaves of the cricket.
  This is not possible since, in both cases, the two leaves are in two adjacent cliques with an anti-matching interaction.
  Therefore, adding the false twin would create a second non-neighbor to the remaining leaf.
\end{proof}
The graph $T_{1,1,1}$ is the claw itself for which \shortindepset is solvable in polynomial time.
The parameterized complexity for the graph $T_{1,1,2}$ (the cricket) remains open.
As a matter of fact, this question is unresolved for $T_{1,1,s}$-free graphs, for any integer $s \geqslant 2$.
Solving those cases would bring us a bit closer to a full dichotomy \emph{FPT vs $W[1]$-hard}.
Although, Theorem~\ref{thm:T122} suggests that this dichotomy will be rather subtle.
In addition, this result infirms the plausible conjecture: \emph{if \shortindepset is FPT in $H$-free graphs, then it is FPT in $H'$-free graphs where $H'$ can be obtained from $H$ by adding false twins}.

The toughest challenge towards the dichotomy is understanding \shortindepset in the absence of \emph{paths of cliques}\footnote{Actually, even the classical complexity of \shortindepset in the absence of long induced paths is not well understood}.
In Theorem~\ref{thm:clique-minus-bipartite}, we make a very first step in that direction: we show that for every graph $H$ with a strong clique decomposition on $P_3$, the problem is FPT.
In the previous paragraphs, we dealt mostly with connected graphs $H$.
In Theorem~\ref{thm:disjoint-union-cliques}, we show that if $H$ is a disjoint union of cliques, then \shortindepset in $H$-free graphs is FPT.
In the language of clique decompositions, this can be phrased as \emph{$H$ has a clique decomposition on an independent set}.

\section{Positive results I: disjoint union of cliques}\label{sec:disjoint-union-cliques}

For $r, q \ge 1$, let $K_r^q$ be the disjoint union of $q$ copies of $K_r$. The proof of the following theorem is inspired by the case $r=2$ by Alekseev \cite{Ale91}.

\begin{theorem}\label{thm:disjoint-union-cliques}
 \indepset is FPT in $K_r^q$-free graphs.
\end{theorem}
\begin{proof}
We will prove by induction on $q$ that a $K_{r}^{q}$-free graph has an independent set of size $k$ or has at most $\Ramsey(r,k)^{qk}n^{qr}$ independent sets. This will give the desired FPT-algorithm, as the proof shows how to construct this collection of independent sets. Note that the case $q=1$ is trivial by Ramsey's theorem. 

Let $G$ be a $K_{r}^{q}$-free graph and let $<$ be any fixed total ordering of $V(G)$. For any vertex $x$, define $x^{+}=\{y,\, x<y\}$ and $x^{-}=V(G)\setminus x^{+}$.

Let $C$ be a fixed clique of size $r$ in $G$ and let $c$ be the smallest vertex of $C$ with respect to $<$. Let $V_1$ be the set of vertices of $c^+$ which have no neighbor in $C$. Note that $V_{1}$ induces a $K_{r}^{q-1}$-free graph, so by induction either it contains an independent set of size $k$, and so does $G$, or it has at most $\Ramsey(r,k)^{(q-1)k}n^{(q-1)r}$ independent sets.
In the latter case, let $\S_{1}$ be the set of all independent sets of $G[V_1]$.% and $\S_1$ the set of all independent sets of $G[V_1 \cup C]$. Note that $|\S_{1}|\leq (r+1)|\mathcal{S'}|$, as an independent set can contain at most one vertex from $C$.

Now in a second phase we define an initially empty set $\S_{C}$ and do the following. For each independent set $S_1$ in $\mathcal S_1$, we denote by $V_2$ the set of vertices in $c^-$ that have no neighbor in $S_1$. 
For every choice of a vertex $x$ amongst the largest $\Ramsey(r,k)$ vertices of $V_2$ in the order, we add $x$ to $S_{1}$ and modify $V_{2}$ in order to keep only vertices that are smaller than $x$ (with respect to $<$) and non adjacent to $x$. We repeat this operation $k$ times (or less if $V_{2}$ becomes empty) and, at the end, we either find an independent set of size $k$ or add $S_{1}$ to $\S_{C}$. By doing so we construct a family of at most $\Ramsey(r,k)^{k}$ independent sets for each $S_{1}$, so in total we get indeed at most $\Ramsey(r,k)^{kq}n^{(q-1)r}$ independent sets for each clique $C$. Finally we define $\cal S$ as the union over all $r$-cliques $C$ of the sets $\S_C$, so that $\cal S$ has size at most the desired number.

We claim that if $G$ does not contain an independent set of size $k$, then $\cal S$ contains all independent sets of $G$. It suffices to prove that for every independent set $S$, there exists a clique $C$ for which $S\in \S_C$.
Let $S$ be an independent set, and define $C$ to be a clique of size $r$ such that its smallest vertex $c$ (with respect to $<$) satisfies the conditions:
\begin{itemize}
\item no vertex of $C$ is adjacent to a vertex of $S \cap c^+$, and
\item $c$ is the smallest vertex such that a clique $C$ satisfying the first item exists.
\end{itemize}
Note that several cliques $C$ might satisfy these conditions. In that case, pick one such clique arbitrarily. These two conditions ensures that $S\cap c^+$ is an independent set in the set $V_1$ defined in the construction above. Thus it will be picked in the second phase as some $S_{1}$ in $\S_{1}$ and for this choice, each time $V_{2}$ is considered, the fact that $C$ is chosen to minimize its smallest element $c$ guarantees that there must be a vertex of $S$ in the $\Ramsey(r,k)$ last vertices in $V_2$, otherwise we could find within those vertices an $r$-clique contradicting the choice of $C$. So we are insured that we will add $S$ to the collection $\S_{C}$, which concludes our proof.
\end{proof}

\section{Positive results II}\label{sec:positive-two}
\subsection{Key ingredient: Iterative expansion and Ramsey extraction}\label{sec:framework}
In this section, we present the main idea of our algorithms. It is a generalization of iterative expansion, which itself is the maximization version of the well-known iterative compression technique.
Iterative compression is a useful tool for designing parameterized algorithms for subset problems (\emph{i.e.} problems where a solution is a subset of some set of elements: vertices of a graph, variables of a logic formula...\emph{etc.})  \cite{CyFoKoLoMaPiPiSa15,ReSmVe04}. Although it has been mainly used for minimization problems, iterative compression has been successfully applied for maximization problems as well, under the name \emph{iterative expansion} \cite{ChLiLuSzZh12}. Roughly speaking, when the problem consists in finding a solution of size at least $k$, the iterative expansion technique consists in solving the problem where a solution $S$ of size $k-1$ is given in the input, in the hope that this solution will imply some structure in the instance. In the following, we consider an extension of this approach where, instead of a single smaller solution, one is given a set of $f(k)$ smaller solutions $S_1$, $\dots$, $S_{f(k)}$. As we will see later, we can further add more constraints on the sets $S_1$, $\dots$, $S_{f(k)}$.
Notice that all the results presented in this sub-section (Lemmas \ref{lem:mis-to-itexp} and \ref{lem:itexp-to-ramsey} in particular) hold for any hereditary graph class (including the class of all graphs). The use of properties inherited from particular graphs (namely, $H$-free graphs in our case) will only appear in Sections~\ref{sec:clique-minus-triangle} and \ref{sec:clique-minus-bipartite}.

\begin{definition}\label{def:fitexp}
For a function $f : \N \rightarrow \N$, the \fitexpmis{$f$} takes as input a graph $G$, an integer $k$, and a set of $f(k)$ independent sets $S_1$, $\dots$, $S_{f(k)}$, each of size $k-1$. The objective is to find an independent set of size $k$ in $G$, or to decide that such an independent set does not exist.
\end{definition}

\begin{lemma}\label{lem:mis-to-itexp}
Let $\G$ be a hereditary graph class. \shortindepset is $FPT$ in $\G$ iff \fitexpmis{$f$} is $FPT$ in $\G$ for some computable function $f : \N \rightarrow \N$.
\end{lemma}
\begin{proof}
Clearly if \shortindepset is $FPT$, then \fitexpmis{$f$} is $FPT$ for any computable function $f$. 
Conversely, let $f$ be a function for which \fitexpmis{$f$} is $FPT$, and let $G$ be a graph with $|V(G)|=n$.

We show by induction on $k$ that there is an algorithm that either finds an independent set of size $k$, or answers that such a set does not exist, in $FPT$ time parameterized by $k$. The initialization can obviously be computed in constant time.
Assume we have an algorithm for $k-1$.
Successively for $i$ from $1$ to $f(k)$, we construct an independent set $S_i$ of size $k-1$ in $G \setminus \left(S_1, \dots, S_{j-1} \right)$. 
If, for some $i$, we are unable to find such an independent set, then it implies that any independent set of size $k$ in $G$ must intersect $S_1 \cup \dots \cup S_i$. We thus branch on every vertex $v$ of this union, and, by induction, find an independent set of size $k-1$ in the graph induced by $V(G) \setminus N[v]$.
If no step $i$ triggered the previous branching, we end up with $f(k)$ vertex-disjoint independent sets $S_1$, $\dots$, $S_{f(k)}$, each of size $k-1$. We now invoke the algorithm for \fitexpmis{$f$} to conclude.
Let us analyze the running time of this algorithm: each step either branch on at most $f(k)(k-1)$ subcases with parameter $k-1$, or concludes in time $\mathcal{A}_{f}(n, k)$, the running time of the algorithm for \fitexpmis{$f$}. Hence the total running time is $O^*(f(k)^k(k-1)^k\mathcal{A}_{f}(n, k))$, where the $O^*(.)$ suppresses polynomial factors.

\end{proof}

We will actually prove a stronger version of this result, by adding more constraints on the input sets $S_1$, $\dots$, $S_{f(k)}$, and show that solving the expansion version on this particular kind of input is enough to obtain the result for \shortindepset.

\begin{definition}
Given a graph $G$ and a set of $k-1$ vertex-disjoint cliques of $G$, $\C = \{C_1, \dots, C_{k-1}\}$, each of size $q$, we say that $\C$ is a set of \emph{Ramsey-extracted cliques of size $q$} if the conditions below hold. Let $C_r = \{c_j^r : j \in \{1, \dots, q\}\}$ for every $r \in \{1, \dots, k-1\}$.
\begin{itemize}
	\item For every $j \in [q]$, the set $\{c_j^r : r \in \{1, \dots, k-1\}\}$ is an independent set of $G$ of size $k-1$.
	\item For any $r \neq r' \in \{1, \dots, k-1\}$, one of the four following case can happen:
	\begin{enumerate}[(i)]
		\item for every $j,j' \in [q]$, $c_j^rc_{j'}^{r'} \notin E(G)$ 
		\item for every $j,j' \in [q]$, $c_j^rc_{j'}^{r'} \in E(G)$ iff $j \neq j'$
		\item for every $j,j' \in [q]$, $c_j^rc_{j'}^{r'} \in E(G)$ iff $j < j'$
		\item for every $j,j' \in [q]$, $c_j^rc_{j'}^{r'} \in E(G)$ iff $j > j'$
	\end{enumerate}
	In the case $(i)$ (resp. $(ii)$), we say that the relation between $C_r$ and $C_{r'}$ is \emph{empty} (resp. \emph{full}\footnote{Remark that in this case, the graph induced by $C_r \cup C_{r'}$ is the complement of a perfect matching.}). In case $(iii)$ or $(iv)$, we say the relation is \emph{semi-full}.
\end{itemize}
\end{definition}

Observe, in particular, that a set $\C$ of $k-1$ Ramsey-extracted cliques of size $q$ can be partitionned into $q$ independent sets of size $k-1$.
As we will see later, these cliques will allow us to obtain more structure with the remaining vertices if the graph is $H$-free. Roughly speaking, if $q$ is large, we will be able to extract from $\C$ another set $\C'$ of $k-1$ Ramsey-extracted cliques of size $q' < q$, such that every clique is a module\footnote{A set of vertices $M$ is a module if every vertex $v \notin M$ is adjacent to either all vertices of $M$, or none.} with respect to the solution $x_1^*$, $\dots$, $x_k^*$ we are looking for. Then, by guessing the structure of the adjacencies between $\C'$ and the solution, we will be able to identify from the remaining vertices $k$ sets $X_1$, $\dots$, $X_k$, where each $X_i$ has the same neighborhood as $x_i^*$ \wrt $\C'$, and plays the role of ``candidates'' for this vertex. For a function $f : \N \rightarrow \N$, we define the following problem:

\begin{definition}\label{def:faugramsey}
The \faugramsey{$f$} problem takes as input an integer $k$ and a graph $G$ whose vertices are partitionned into non-empty sets $X_1 \cup \dots \cup X_k \cup C_1 \cup \dots \cup C_{k-1}$, where:
\begin{itemize}
	\item  $\{C_1, \dots, C_{k-1}\}$ is a set of $k-1$ Ramsey-extracted cliques of size $f(k)$
	\item any independent set of size $k$ in $G$ is contained in $X_1 \cup \dots \cup X_k$ 
	%\item if $G$ has an independent set of size $k$, then there is one which has a non-empty intersection with $X_i$, for every $i \in \{1, \dots, k\}$
	\item $\forall i \in \{1, \dots, k\}$, $\forall v,w \in X_i$ and $\forall j \in \{1, \dots, k-1\}$, $N(v) \cap C_j=N(w) \cap C_j=\emptyset$ or $N(v) \cap C_j =N(w) \cap C_j=C_j$
	\item the following bipartite graph $\B$ is connected: $V(\B) = B_1 \cup B_2$, $B_1=\{b_1^1, \dots, b_k^1\}$, $B_2=\{b_1^2, \dots, b_{k-1}^2\}$ and $b_j^1b_r^2 \in E(\B)$ iff $X_j$ and $C_r$ are adjacent.
	\end{itemize}
%The objective is to find an independent set $S$ in $G$ of size at least $k$ such that $S \cap X_i \neq \emptyset$ for all $i \in \{1, \dots, k\}$, or to decide that such an independent set does not exist.
The objective is to find an independent set $S$ in $G$ of size at least $k$, or to decide that $G$ does not contain an independent set $S$ such that $S \cap X_i \neq \emptyset$ for all $i \in \{1, \dots, k\}$.
\end{definition}

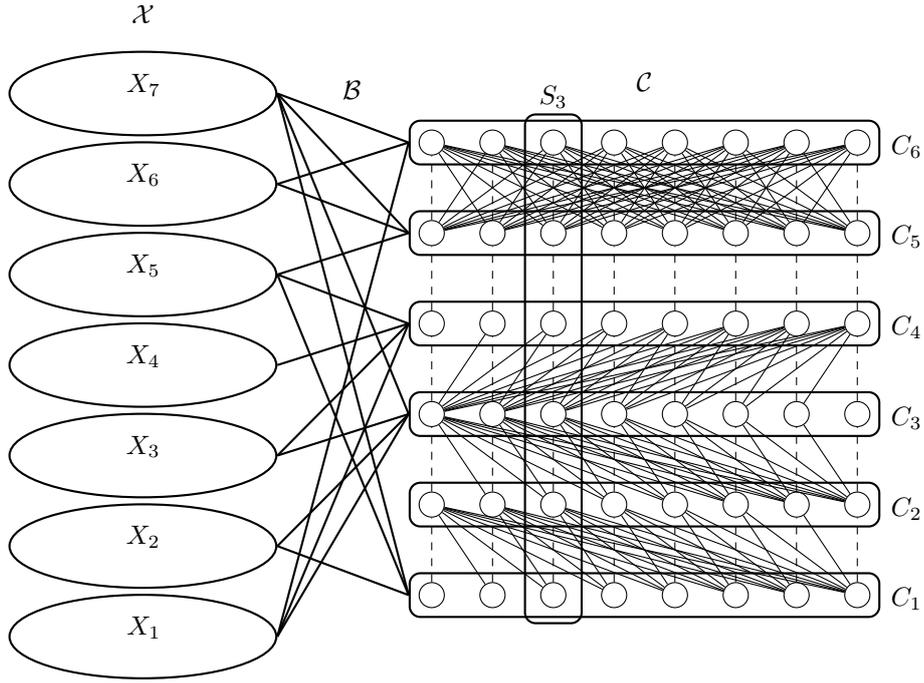
\begin{figure}
  \centering
  \begin{tikzpicture}
    \def\k{7}
    \def\q{8}
    \def\hs{0.8}
    \def\vs{1.2}
    \pgfmathsetmacro\km{\k-1}
    \pgfmathsetmacro\kmm{\km-1}
    \pgfmathsetmacro\qm{\q-1}

    %X
    \node at (-3,\vs * \k + 0.5) {$\X$} ;
    \foreach \i in {1,...,\k}{
      \node (a\i) at (-4,\vs * \i - \hs / 2 - 0.3) {} ;
      \node (b\i) at (-2,\vs * \i - \hs / 2 - 0.3) {} ;
      \node (c\i) at (-3,\vs * \i - \hs / 2) {} ;
      \node[draw,ellipse, thick, fit=(a\i) (b\i) (c\i)] (X\i) {$X_\i$} ;
    }

    %C
    \node at (\hs * \q / 2 + \hs / 2,\vs * \km + 0.8) {$\C$} ;
    \foreach \i in {1,...,\km}{
      \foreach \j in {1,...,\q}{
        \node[draw,circle] (u\i\j) at (\hs * \j,\vs * \i) {} ;
      }
      \node[draw,rectangle,rounded corners, thick, fit=(u\i1) (u\i\q)] (C\i) {} ;
      \node at (\hs * \q + \hs - 0.15,\vs * \i - 0.05) {$C_\i$} ; 
    }
    \foreach \i [count=\ip from 2] in {1,...,\kmm}{
      \foreach \j in {1,...,\q}{
        \draw[dashed] (u\i\j) -- (u\ip\j) ;
      }
    }
    \node[draw,rectangle,rounded corners,inner sep=0.2cm, thick, fit=(u13) (u63)] (S3) {} ;
    \node at (3 * \hs, \k * \vs - 0.6) {$S_3$} ;

    %examples of interaction between X and C
    \node at (-0.25, \k * \vs - 0.5) {$\B$} ;
    \foreach \i/\j in {1/3,1/4,1/6,2/1,2/3,3/3,3/4,4/4,5/1,5/4,5/5,6/5,6/6,7/1,7/3,7/5,7/6}{
      \draw[thick] (X\i.east) -- (C\j.west) ;
    }
    
    %examples of interaction in C
    %half-graphs
    \foreach \a/\b in {3/4,3/2,2/1}{
      \foreach \j in {1,...,\qm}{ 
        \pgfmathsetmacro\jp{\j+1}
        \foreach \h in {\q,...,\jp}{
          \draw (u\a\j) -- (u\b\h) ;
        }
      }
    }
    %antimatchings
    \foreach \a/\b in {5/6}{
      \foreach \j in {2,...,\qm}{ 
        \pgfmathsetmacro\jp{\j+1}
        \pgfmathsetmacro\jm{\j-1}
        \foreach \h in {1,...,\jm}{
          \draw (u\a\j) -- (u\b\h) ;
        }
        \foreach \h in {\q,...,\jp}{
          \draw (u\a\j) -- (u\b\h) ;
        }
      }
      \foreach \h in {2,...,\q}{
          \draw (u\a1) -- (u\b\h) ;
      }
      \foreach \h in {1,...,\qm}{
          \draw (u\a\q) -- (u\b\h) ;
      }
    }
  \end{tikzpicture}
  \caption{The structure of the \faugramsey{$f$} inputs.}
  \label{fig:faugramsey}
\end{figure}

%\edouard{On pourrait même rajouter dans l'input que les $X_i$ voient tout ou rien de chaque clique extraite}

%As the name of the problem suggests, the algorithm will heavily rely on Ramsey's theorem, which states that every graph on $\Ramsey(r, k)$ contains either a clique of size $r$ or an independent set of size $k$. We will also use the colored version of the theorem, which states that a complete graph whose edges are colored with $c$ colors contains a monochromatic clique if its number of vertices is at least $\Ramsey(k, \dots, k)$.

\begin{lemma}\label{lem:itexp-to-ramsey}
Let $\G$ be a hereditary graph class. If there exists a computable function $f : \N \rightarrow \N$ such that \faugramsey{$f$} is $FPT$ in $\G$, then \fitexpmis{$g$} is $FPT$ in $\G$, where $g(x)=\Ramsey_{\ell}(f(x)2^{x(x-1)})~\forall x \in \N$, with $\ell_x=2^{(x-1)^2}$.
\end{lemma}
\begin{proof}
Let $f : \N \rightarrow \N$ be such a function, and let $G$, $k$ and $\S = \{S_1, \dots, S_{g(k)}\}$ be an input of \fitexpmis{$g$}.
Recall that the objective is to find an independent set of size $k$ in $G$, or to decide that such an independent set does not exist.
If $G$ contains an independent set of size $k$, then either there is one intersecting some set of $\S$, or every independent set of size $k$ avoids the sets in $\S$. In order to capture the first case, we branch on every vertex $v$ of the sets in $\S$, and make a recursive call with parameter $G \setminus N[v]$, $k-1$.
In the remainder of the algorithm, we thus assume that any independent set of size $k$ in $G$ avoids every set of $\S$.

We choose an arbitrary ordering of the vertices of each $S_j$. Let us denote by $s_j^r$ the $r^{th}$ vertex of $S_j$. Notice that given an ordered pair of sets of $k-1$ vertices $(A, B)$, there are $\ell_k=2^{(k-1)^2}$ possible sets of edges between these two sets. Let us denote by $c_1$, $\dots$, $c_{2^{(k-1)^2}}$ the possible sets of edges, called \emph{types}.
We define an auxiliary edge-colored graph $H$ whose vertices are in one-to-one correspondence with $S_1$, $\dots$, $S_{g(k)}$, and, for $i < j$, there is an edge between $S_i$ and $S_j$ of color $\gamma$ iff the type of $(S_i,S_j)$ is $\gamma$. 
By Ramsey's theorem, since $H$ has $\Ramsey_{\ell_k}(f(k)2^{k(k-1)})$ vertices, it must admit a monochromatic clique of size at least $h(k)=f(k)2^{k(k-1)}$. \emph{W.l.o.g.}, the vertex set of this clique corresponds to $S_1$, $\dots$, $S_{h(k)}$. For $p \in \{1, \dots, k-1\}$, let $C_p = \{s_j^p, \dots, s_{h(k)}^p\}$. 
Observe that the Ramsey extraction ensures that each $C_p$ is either a clique or an independent set. If $C_p$ is an independent set for some $r$, then we can immediately conclude, since $h(k) \ge k$.
Hence, we suppose that $C_p$ is a clique for every $p \in \{1, \dots, k-1\}$. 
We now prove that $C_1$, $\dots$, $C_{k-1}$ are Ramsey-extracted cliques of size $h(k)$.
First, by construction, for every $j \in \{1, \dots, h(k)\}$, the set $\{s_j^p : p=1, \dots, k-1\}$ is an independent set. 
Then, let $c$ be the type of the clique obtained previously, represented by the adjacencies between two sets $(A, B)$, each of size $k-1$. For every $p \in \{1, \dots, k-1\}$, let $a_p$ (resp. $b_p$) be the $a^{th}$ vertex of $A$ (resp. $B$). Let $p,q \in \{1, \dots, t\}$, $p \neq q$.
If any of $a_pb_q$ and $a_qb_p$ are edges in type $c$, then there is no edge between $C_p$ and $C_q$, and their relation is thus empty.
If both edges $a_pb_q$ and $a_qb_p$ exist in $c$, then the relation between $C_p$ and $C_q$ is full.
Finally if exactly one edge among $a_pb_q$ and $a_qb_p$ exists in $c$, then the relation between $C_p$ and $C_q$ is semi-full.
This concludes the fact that $\C=\{C_1, \dots, C_{k-1}\}$ are Ramsey-extracted cliques of size $h(k)$.
%In the following, we say that two cliques $C_r$ and $C_t$ are adjacent if they are either full or semi-full.

Suppose that $G$ has an independent set $X^*=\{x_1^*, \dots, x_k^*\}$. Recall that we assumed previously that $X^*$ is contained in $V(G) \setminus \left(C_1 \cup \dots \cup C_{k-1}\right)$.
The next step of the algorithm consists in branching on every subset of $f(k)$ indices $J \subseteq \{1, \dots, h(k)\}$, and restrict every set $C_p$ to $\{s_j^p : j \in J\}$. For the sake of readability, we keep the notation $C_p$ to denote $\{s_j^p : j \in J\}$ (the non-selected vertices are put back in the set of remaining vertices of the graph, \ie we do not delete them).
Since $h(k) = f(k)2^{k(k-1)}$, there must exist a branching where the chosen indices are such that for every $i \in \{1, \dots, k\}$ and every $p \in \{1, \dots, k-1\}$, $x_i^*$ is either adjacent to all vertices of $C_p$ or none of them. 
In the remainder, we may thus assume that such a branching has been made, with respect to the considered solution $X^*=\{x_1^*, \dots, x_k^*\}$.
Now, for every $v \in V(G) \setminus \left(C_1, \dots, C_{k-1}\right)$, if there exists $p \in \{1, \dots, k-1\}$ such that $N(v) \cap C_p \neq \emptyset$ and $N(v) \cap C_p \neq C_p$  , then we can remove this vertex, as we know that it cannot correspond to any $x_i^*$. Thus, we know that all the remaining vertices $v$ are such that for every $p \in \{1, \dots, k-1\}$, $v$ is either adjacent to all vertices of $C_p$, or none of them.

In the following, we perform a color coding-based step on the remaining vertices. Informally, this color coding will allow us to identify, for every vertex $x_i^*$ of the optimal solution, a set $X_i$ of candidates, with the property that all vertices in $X_i$ have the same neighborhood with respect to sets $C_1$, $\dots$, $C_{k-1}$.
We thus color uniformly at random the remaining vertices $V(G) \setminus \left(C_1, \dots, C_{k-1}\right)$ using $k$ colors. The probability that the elements of $X^*$ are colored with pairwise distinct colors is at least $e^{-k}$.
We are thus reduced to the case of finding a \emph{colorful}\footnote{A set of vertices is called \emph{colorful} if it is colored with pairwise distinct colors.} independent set of size $k$.
For every $i \in \{1, \dots, k\}$, let $X_i$ be the vertices of $V(G) \setminus \left(C_1, \dots, C_{k-1}\right)$ colored with color $i$.
We now partition every set $X_i$ into at most $2^{k-1}$ subsets $X_i^1$, $\dots$, $X_i^{2^{k-1}}$, such that for every $j \in \{1, \dots, 2^{k-1}\}$, all vertices of $X_i^j$ have the same neighborhood with respect to the sets $C_1$, $\dots$, $C_{k-1}$ (recall that every vertex of $V(G) \setminus \left(C_1, \dots, C_{k-1}\right)$ is adjacent to all vertices of $C_p$ or none, for each $p \in \{1, \dots, k-1\}$).
We branch on every tuple $(j_1, \dots, j_k) \in \{1, \dots, 2^{k-1}\}$. Clearly the number of branchings is bounded by a function of $k$ only and, moreover, one branching $(j_1, \dots, j_k)$ is such that $x_i^*$ has the same neighborhood in $C_1 \cup \dots \cup C_{k-1}$ as vertices of $X_i^{j_i}$ for every $i \in \{1, \dots, k\}$. We assume in the following that such a branching has been made. For every $i \in \{1, \dots, k\}$, we can thus remove vertices of $X_i^j$ for every $j \neq j_i$. For the sake of readability, we rename $X_i^{j_i}$ as $X_i$.
Let $\B$ be the bipartite graph with vertex bipartition $(B_1, B_2)$, $B_1=\{b_1^1, \dots, b_k^1\}$, $B_2=\{b_1^2, \dots, b_{k-1}^2\}$, and $b^1_ib^2_p \in E(\B)$ iff $x_i^*$ is adjacent to $C_p$.
Since every $x_i^*$ has the same neighborhood as $X_i$ with respect to $C_1$, $\dots$, $C_{k-1}$, this bipartite graph actually corresponds to the one described in Definition~\ref{def:faugramsey} representing the adjacencies between $X_i$'s and $C_p$'s.
We now prove that it is connected. Suppose it is not. Then, since $|B_1|=k$ and $|B_2|=k-1$, there must be a component with as many vertices from $B_1$ as vertices from $B_2$. However, in this case, using the fixed solution $X^*$ on one side and an independent set of size $k-1$ in $C_1 \cup \dots \cup C_{k-1}$ on the other side, it implies that there is an independent set of size $k$ intersecting $\cup_{p=1}^{k-1}C_p$, a contradiction.

Hence, all conditions of Definition~\ref{def:faugramsey} are now fulfilled.
It now remains to find an independent set of size $k$ disjoint from the sets $\C$, and having a non-empty intersection with $X_i$, for every $i \in \{1, \dots, k\}$. We thus run an algorithm solving \faugramsey{$f$} on this input, which concludes the algorithm.
\end{proof}

The proof of the following result is immediate, by using successively Lemmas~\ref{lem:mis-to-itexp} and \ref{lem:itexp-to-ramsey}.
\begin{theorem}\label{thm:iterative-ramsey}
Let $\G$ be a hereditary graph class. If \faugramsey{$f$} is $FPT$ in $\G$ for some computable function $f$, then \shortindepset is $FPT$ in $\G$.
\end{theorem}

We now apply this framework to two families of graphs $H$.

\subsection{Clique minus a smaller clique}\label{sec:clique-minus-triangle}

\begin{theorem}
For any $r \ge 2$ and $s < r$, \shortindepset in $(K_r \setminus K_s)$-free graphs is $FPT$ if $s \le 3$, and $W[1]$-hard otherwise.
\end{theorem}
\begin{proof}
The case $s=2$ was already known~\cite{Dabrowski12}. The result for $s \ge 4$ comes from Theorem~\ref{thm:Whard}. We now deal with the case $s=3$. We solve the problem in $(K_{r+3} \setminus K_3)$-free graphs, for every $r \ge 2$ (the problem is polynomial for $r=1$, since it it corresponds exactly to the case of claw-free graphs). 
Let $G, k$ be an input of the problem.
We present an $FPT$ algorithm for \faugramsey{$f$} with $f(x)=r$ for every $x \in \N$. The result for \shortindepset can then be obtained using Theorem~\ref{thm:iterative-ramsey}. 

We thus assume that $V(G) = X_1 \cup \dots  \cup X_k \cup C_1 \cup \dots \cup C_{k-1}$ where all cliques $C_p$ have size $r$.
Consider the bipartite graph $\B$ representing the adjacencies between $\{X_1, \dots, X_k\}$ and $\{C_1, \dots, C_{k-1}\}$, as in Definition~\ref{def:faugramsey} (for the sake of readability, we will make no distinction between the vertices of $\B$ and the sets $\{X_1, \dots, X_k\}$ and $\{C_1, \dots, C_{k-1}\}$). 
We may first assume that $|X_i| \ge \Ramsey(r, k)$ for every $i \in \{1, \dots, k\}$, since otherwise we can branch on every vertex $v$ of $X_i$ and make a recursive call with input $G \setminus N[v]$, $k-1$.
Suppose that $G$ contains an independent set $S^*=\{x_1^*, \dots, x_k^*\}$, with $x_i \in X_i$ for all $i \in \{1, \dots, k\}$.
The first step is to consider the structure of $\B$, using the fact that $G$ is $(K_r \setminus K_3)$-free. We have the following:
\begin{claim}
$\B$ is a path.
\end{claim}
\begin{claimproof}
We first prove that for every $i \in \{1, \dots, k\}$, the degree of $X_i$ in $\B$ is at most $2$. Indeed, assume by contradiction that it is adjacent to $C_a$, $C_b$ and $C_c$. Since $|X_i| \ge \Ramsey(r, k)$, by Ramsey's theorem, it either contains an independent set of size $k$, in which case we are done, or a clique $K$ of size $r$. However, observe in this case that $K$ together with $s_1^a$, $s_1^b$ and $s_1^c$ (which are pairwise non-adjacent) induces a graph isomorphic to $K_{r+3} \setminus K_3$.

Then, we show that for every $i \in \{1, \dots, k-1\}$, the degree of $C_i$ in $\B$ is at most $2$. Assume by contradiction that $C_i$ is adjacent to $X_a$, $X_b$ and $X_c$. If the instance is positive, then there must be an independent set of size three with non-empty intersection with each of $X_a$, $X_b$ and $X_c$. If such an independent set does not exist (which can be checked in cubic time), we can immediately answer NO. Now observe that $C_i$ (which is of size $r$) together with this independent set induces a graph isomorphic to $K_{r+3} \setminus K_3$.

To summarize, $\B$ is a connected bipartite graph of maximum degree $2$ with $k$ vertices in one part, $k-1$ vertices in the other part. It must be a path.
\end{claimproof}

W.l.o.g., we may assume that for every $i \in \{2, \dots, k-1\}$, $X_i$ is adjacent to $C_{i-1}$ and $C_i$, and that $X_1$ (resp. $X_k$) is adjacent to $C_1$ (resp. $C_{k-1}$). We now concentrate on the adjacencies between sets $X_i$'s.
We say that an edge $xy \in E(G)$ is a \emph{long edge} if $x \in X_i$, $y \in X_j$ with $|j-i| \ge 2$ and $2 \le i,j \le k-1$, $i \neq j$. 

\begin{claim}
$\forall x \in X_2 \cup \dots \cup X_{k-1}$, $x$ is incident to at most $(k-2)(\Ramsey(r, 3)-1)$ long edges.
%If $|j-i| \ge 2$ and $j \neq 1, k$, every vertex of $X_i$ has at most $\Ramsey(r)-1$ neighbors in $X_j$.
%If $|j-i| \ge 2$ and $j \neq 1, k$, then $\forall x \in X_i$, $|N(x) \cap X_j| \le \Ramsey(r, 3)-1$.
\end{claim}
\begin{claimproof}
To do so, for $i, j \in \{2, \dots, k-1\}$ such that $|j-i| \ge 2$, $i \neq j$, we prove that $\forall x \in X_i$, $|N(x) \cap X_j| \le \Ramsey(r, 3)-1$. Assume by contradiction that $x \in X_i$ has at least $\Ramsey(r, 3)$ neighbors $Y \subseteq X_j$. By Ramsey's theorem, either $Y$ contains an independent set of size $3$ or a clique of size $r$. In the first case, $C_j$ together with these three vertices induces a graph isomorphic to $K_{r+3} \setminus K_3$. Hence we may assume that $Y$ contains a clique $Y'$ of size $r$. But in this case, $Y'$ together with $x$, $s_1^{j-1}$, $s_1^j$ induce a graph isomorphic to $K_{r+3} \setminus K_3$ as well.
\end{claimproof}

Recall that the objective is to find an independent set of size $k$ with non-empty intersection with $X_i$, for every $i \in \{1, \dots, k\}$. 
We assume $k \ge 5$, otherwise the problem is polynomial. The algorithm starts by branching on every pair of non-adjacent vertices $(x_1, x_k) \in X_1 \times X_k$, and removing the union of their neihborhoods in $X_2 \cup \dots \cup X_{k-1}$. For the sake of readability, we still denote by $X_2$, $\dots$, $X_{k-1}$ these reduced sets. If such a pair does not exist or the removal of their neighborhood empties some $X_i$, then we immediately answer NO (for this branch). Informally speaking, we just guessed the solution within $X_1$ and $X_k$ (the reason for this is that we cannot bound the number of long edges incident to vertices of these sets).
We now concentrate on the graph $G'$, which is the graph induced by $X_2 \cup \dots \cup X_{k-1}$. Clearly, it remains to decide whether $G'$ admits an independent set of size $k-2$ with non-empty intersection with $X_i$, for every $i \in \{2, \dots, k-1\}$.

The previous claim showed that the structure of $G'$ is quite particular: roughly speaking, the adjacencies between consecutive $X_i$'s is arbitrary, but the number of long edges is bounded for every vertex.
The key observation is that if there were no long edge at all, then a simple dynamic programming algorithm would allow us to conclude.
Nevertheless, using the previous claim, we can actually upper bound the number of long edges incident to a vertex of the solution by a function of $k$ only (recall that $r$ is a constant). We can then get rid of these problematic long edges using the so-called technique of \emph{random separation}~\cite{CaChCh06}.
Let $S= \{x_2, \dots, x_{k-1}\}$ be a solution of our problem (with $x_i \in X_i$ for every $i \in \{2, \dots, k-1\}$).
Let us define $D = \{y: xy$ is a long edge and $x \in S\}$. By the previous claim, we have $|D| \le (\Ramsey(r, 3)-1)(k-2)^2$. The idea of random separation is to delete each vertex of the graph with probability $\frac{1}{2}$. At the end, we say that a removal is \emph{successful} if both of the two following conditions hold: (i) no vertex of $S$ has been removed, and (ii) all vertices of $D$ have been removed (other vertices but $S$ may have also been removed). Observe that the probability that a removal is successful is at least $2^{-k^2\Ramsey(r, 3)}$. In such a case, we can remove all remaining long edges: indeed, for a remaining long edge $xy$, we know that there exists a solution avoiding both $x$ and $y$, hence we can safely delete $x$ and $y$. As previously, we still denote by $X_2$, $\dots$, $X_{k-1}$ the reduced sets, for the sake of readability.
We thus end up with a graph composed of sets $X_2$, $\dots$, $X_{k-1}$, with edges between $X_i$ and $X_j$ only if $[j-i| = 1$. In that case, observe that there is a solution if and only if the following dynamic programming returns $true$ on input $P(3, x_2)$ for some $x_2 \in X_2$:

$$
P(i,x_{i-1}) = \left\{
    \begin{array}{ll}
    	true & \mbox{ if } i=k \\
        false & \mbox{ if } X_i \subseteq   N(x_{i-1}) \\
        \bigvee_{x_i \in X_i \setminus N(x_{i-1})} P(i+1, x_i) & \mbox{ otherwise.}
    \end{array}
\right.
$$

Clearly this dynamic programming runs in $O(mnk)$ time, where $m$ and $n$ are the number of edges and vertices of the remaining graph, respectively. Moreover, it can easily be turned into an algorithm returning a solution of size $k-2$ if it exists.

Finally, similarly to classical random separation algorithms, it is sufficient to repeat this process $O(2^{k^2\Ramsey(r)})$ times in order to obtain an $FPT$ one-sided error Monte Carlo algorithm with constant success probability. Moreover, such an algorithm can be derandomized up to an additional $2^kk^{O(\log k)}$ factor in the running time \cite{CyFoKoLoMaPiPiSa15}.

\end{proof}

\subsection{Clique minus a complete bipartite graph}\label{sec:clique-minus-bipartite}

%\remi{oublier cette notation, utiliser plutôt $K_r \setminus K_{s_1, s_2}$}
For every three positive integers $r$, $s_1$, $s_2$ with $s_1 + s_2 < r$, we consider the graph $K_r \setminus K_{s_1, s_2}$.
Another way to see $K_r \setminus K_{s_1, s_2}$ is as a \emph{$P_3$ of cliques} of size $s_1$, $r-s_1-s_2$, and $s_2$.
More formally, every graph $K_r \setminus K_{s_1, s_2}$ can be obtained from a $P_3$ by adding $s_1-1$ false twins of the first vertex, $r-s_1-s_2-1$, for the second, and $s_2-1$, for the third.
 
 \begin{theorem}\label{thm:clique-minus-bipartite}
 For any $r \ge 2$ and $s_1 \le s_2$ with $s_1+s_2 < r$, \shortindepset in $K_r \setminus K_{s_1, s_2}$-free graphs is FPT. 
 \end{theorem}

 \begin{proof}
   %\remi{Est-ce que ca ne simplifie pas de considerer $K_{3r} \setminus K_{r, r}$, i.e. quand les trois cliques sont de meme taille? ca implique l'autre resultat en prenant $r=max \{s_1, s_2, r-s_1-s_2\}$}
   It is more convenient to prove the result for $K_{3r} \setminus K_{r, r}$-free graphs, for any positive integer $r$.
   It implies the theorem by choosing this new $r$ to be larger than $s_1$, $s_2$, and $r-s_1-s_2$.
   We will show that for $f(x) := 3r$ for every $x \in \N$, \faugramsey{$f$} in $K_{3r} \setminus K_{r, r}$-free graphs is FPT.
   By Theorem~\ref{thm:iterative-ramsey}, this implies that \shortindepset is FPT in this class.
   Let $C_1,\ldots, C_{k-1}$ (whose union is denoted by $\C$) be the Ramsey-extracted cliques of size $3r$, which can be partitionned, as in Definition~\ref{def:faugramsey}, into $3r$ independent sets
   $S_1, \ldots, S_{3r}$, each of size $k-1$. Let $\X=\bigcup_{i=1}^k X_i$ be the set in which we are looking for an independent set of size $k$.
   We recall that between any $X_i$ and any $C_j$ there are either all the edges or none.
   Hence, the whole interaction between $\X$ and $\C$ can be described by the bipartite graph $\B$ described in Definition~\ref{def:faugramsey}.
   Firstly, we can assume that each $X_i$ is of size at least $\Ramsey(r, k)$, otherwise we can branch on $\Ramsey(r, k)$ choices to find one vertex in an optimum solution.
 By Ramsey's theorem, we can assume that each $X_i$ contains a clique of size $r$ (if it contains an independent set of size $k$, we are done).
   %\edouard{perhaps, we could move this type of restriction to Theorem~\ref{thm:iterative-ramsey}?}
   Our general strategy is to leverage the fact that the input graph is $(K_{3r} \setminus K_{r, r})$-free to describe the structure of $\X$.
   Hopefully, this structure will be sufficient to solve our problem in FPT time.

   We define an auxiliary graph $Y$ with $k-1$ vertices.
   The vertices $y_1, \ldots, y_{k-1}$ of $Y$ represent the Ramsey-extracted cliques of $\C$ and two vertices $y_i$ and $y_j$ are adjacent iff the relation between $C_i$ and $C_j$ is not empty (equivalently the relation is full or semi-full).
   It might seem peculiar that we concentrate the structure of $\C$, when we will eventually discard it from the graph.
   It is an indirect move: the simple structure of $\C$ will imply that the interaction between $\X$ and $\C$ is simple, which in turn, will severely restrict the subgraph induced by $\X$.
   More concretely, in the rest of the proof, we will (1) show that $Y$ is a clique, (2) deduce that $\B$ is a complete bipartite graph, (3) conclude that $\X$ cannot contain an induced $K^2_r = K_r \uplus K_r$ and run the algorithm of Theorem~\ref{thm:disjoint-union-cliques}.

   Suppose that there is $y_{i_1}y_{i_2}y_{i_3}$ an induced $P_3$ in $Y$, and consider $C_{i_1}$, $C_{i_2}$, $C_{i_3}$ the corresponding Ramsey-extracted cliques.
   For $s<t \in [3r]$, let $C_i^{s \rightarrow t} := C_i \cap \bigcup_{s \leqslant j \leqslant t} S_j$. In other words, $C_i^{s \rightarrow t}$ contains the elements of $C_i$ having indices between $s$ and $t$. Since $|C_i|=3r$, each $C_i$ can be partitionned into three sets, of $r$ elements each: $C_{i}^{1 \rightarrow r}$, $C_{i}^{r+1 \rightarrow 2r}$ and $C_{i}^{2r+1 \rightarrow 3r}$.
   Recall that the relation between $C_{i_1}$ and $C_{i_2}$ (resp. $C_{i_2}$ and $C_{i_3}$) is either full or semi-full, while the relation between $C_{i_1}$ and $C_{i_3}$ is empty.
   This implies that at least one of the four following sets induces a graph isomorphic to $K_{3r} \setminus K_{r, r}$:
   \begin{itemize} 
   		\item $C_{i_1}^{1 \rightarrow r} \cup C_{i_2}^{r+1 \rightarrow 2r} \cup C_{i_3}^{1 \rightarrow r}$
   		\item $C_{i_1}^{1 \rightarrow r} \cup C_{i_2}^{r+1 \rightarrow 2r} \cup C_{i_3}^{2r+1 \rightarrow 3r}$
   		\item $C_{i_1}^{2r+1 \rightarrow 3r} \cup C_{i_2}^{r+1 \rightarrow 2r} \cup C_{i_3}^{1 \rightarrow r}$
   		\item $C_{i_1}^{2r+1 \rightarrow 3r} \cup C_{i_2}^{r+1 \rightarrow 2r} \cup C_{i_3}^{2r+1 \rightarrow 3r}$
   \end{itemize}
   Hence, $Y$ is a disjoint union of cliques.
   Let us assume that $Y$ is the union of at least two (maximal) cliques.
   
   Recall that the bipartite graph $\B$ is connected.
   Thus there is $b^1_h \in B_1$ (corresponding to $X_h$) adjacent to $b^2_i \in B_2$ and $b^2_j \in B_2$ (corresponding to $C_i$ and $C_j$, respectively), such that $y_i$ and $y_j$ lie in two different connected components of $Y$ (in particular, the relation between $C_i$ and $C_j$ is empty).
   Recall that $X_h$ contains a clique of size at least $r$.
   This clique induces, together with any $r$ vertices in $C_i$ and any $r$ vertices in $C_j$, a graph isomorphic to $K_{3r} \setminus K_{r, r}$; a contradiction.
   Hence, $Y$ is a clique.

   Now, we can show that $\B$ is a complete bipartite graph.
   Each $X_h$ has to be adjacent to at least one $C_i$ (otherwise this trivially contradicts the connectedness of $\B$).
   If $X_h$ is not linked to $C_j$ for some $j \in \{1, \dots, k-1\}$, then a clique of size $r$ in $X_h$ (which always exists) induces, together with $C_i^{1 \rightarrow r} \cup C_j^{2r+1 \rightarrow 3r}$ or with $C_i^{2r+1 \rightarrow 3r} \cup C_j^{1 \rightarrow r}$, a graph isomorphic to $K_{3r} \setminus K_{r, r}$.

   Since $\B$ is a complete bipartite graph, every vertex of $C_1$ dominates all vertices of $\X$
   In particular, $\X$ is in the intersection of the neighborhood of the vertices of some clique of size $r$.
   This implies that the subgraph induced by $\X$ is $(K_r \uplus K_r)$-free.
   Hence, we can run the FPT algorithm of Theorem~\ref{thm:disjoint-union-cliques} on this graph.
 \end{proof}

\subsection{The gem}
Let the \emph{gem} be the graph obtained by adding a universal vertex to a path on four vertices (see Figure \ref{fig:gem}). Using our framework once again, we are able to obtain the following result:

\begin{center}
\begin{figure}
\centering
 \begin{tikzpicture} [scale=1]
 \node[draw, circle, fill=black, inner sep=-0.05cm] at (0, 0.5) {} ;
 \node[draw, circle, fill=black, inner sep=-0.05cm] at (0.5, 1) {} ;
 \node[draw, circle, fill=black, inner sep=-0.05cm] at (1, 0) {} ;
 \node[draw, circle, fill=black, inner sep=-0.05cm] at (1.5, 1) {} ;
 \node[draw, circle, fill=black, inner sep=-0.05cm] at (2,0.5) {} ;
 \draw (1, 0) -- (0,0.5);
 \draw (1, 0) -- (0.5, 1);
 \draw (1, 0) -- (1.5, 1);
 \draw (1, 0) -- (2, 0.5);
 \draw (0, 0.5) -- (0.5, 1);
 \draw (0.5, 1) -- (1.5, 1);
 \draw (1.5, 1) -- (2, 0.5);
  
   \end{tikzpicture}
   \caption{The gem.}
   \label{fig:gem}
\end{figure}
\end{center}

\begin{theorem}
There is a randomized $FPT$ algorithm for \shortindepset in $gem$-free graphs.
\end{theorem}
\begin{proof}
Let $f(x) := 1$ for every $x \in \mathbb{N}$. We prove that \faugramsey{$f$} admits a randomized $FPT$ in $gem$-free graphs. By the definition of $f$, we have $C_p = \{c_p\}$ for every $p \in \{1, \dots, k-1\}$. Recall that the objective is to find a rainbow independent set in $G$, or to decide that a 
$\alpha(G) < k$.
 Since the bipartite graph $\B$ representing the adjacencies between $\{X_1, \dots, X_k\}$ and $\{c_1, \dots, c_{k-1}\}$ is connected, it implies that for every $i \in \{1, \dots, k\}$, there exists $p \in \{1, \dots, k-1\}$ such that $c_p$ dominates all vertices of $X_i$. Since $G$ is $gem$-free, it implies that $G[X_i]$ is $P_4$-free for every $i \in \{1, \dots, k\}$. Since $P_4$-free graphs (\textit{a.k.a} cographs) are perfect, the size of a maximum independent set equals the size of a clique cover. 
If $G[X_i]$ contains an independent set of size $k$ (which can be tested in polynomial time), then we are done. Otherwise, we can, still in polynomial time, partition the vertices of $X_i$ into at most $k-1$ sets $X_i^1$, $\dots$, $X_i^{q_i}$, where $G[X_i^j]$ induces a clique for every $j \in \{1, \dots, q_i\}$.
We now perform a branching for every tuple $(j_1, \dots, j_k)$, where $j_i \in \{1, \dots, q_i\}$ for every $i \in \{1, \dots, k\}$, which, informally, allows us to guess the clique $X_i^{j_i}$ which contains the element of the rainbow independent set we are looking for. For the sake of readability, we allow ourselves this slight abuse of notation: we rename $X_i^{j_i}$ into simply $X_i$.
Thus, for every $i \in \{1, \dots, k\}$, $G[X_i]$ is a clique.

Now, let $i, j \in \{1, \dots, k\}$, $i \neq j$. Let us analyse the adjacencies between $X_i$ and $X_j$. We say that $\{a, b, c, d\} \subseteq X_i \cup X_j$ is a \emph{balanced diamond} if $a, b \in X_i$ ($a \neq b$), $c, d \in X_j$ ($c \neq d$) and all vertices $\{a, b, c, d\}$ are pairwise adjacent but $\{b, d\}$. We have the following claim:
\begin{claim}\label{claim:diamondtwins}
If the graph induced by $X_i \cup X_j$ has a balanced diamond, then $X_i$ and $X_j$ are twins in $\B$.
\end{claim}
\begin{claimproof}
Suppose they are not. \textit{W.l.o.g.} we assume that $X_i$ is adjacent to $\{c_p\}$ while $X_j$ is not, for some $p \in \{1, \dots, k-1\}$. Then the vertices of the balanced diamond together with $c_p$ induce a $gem$.
\end{claimproof}

The remainder of the proof consists of ``cleaning'' the adjacencies $(X_i, X_j)$ having no balanced diamond. In that case, observe that $X_i$ and $X_j$ can respectively be partitioned into $X_i^0$, $X_i^1$, $\dots$, $X_i^{q}$ and $X_j^0$, $X_j^1$, $\dots$, $X_j^{q}$ (where $X_i^0$ and $X_j^0$ are potentially empty) such that $X_i^r \cup X_j^r$ induces a clique for every $r \in \{1, \dots, q\}$, and there is no edge between $X_i^r$ and $X_j^{r'}$ whenever $r \neq r'$ or $r=0$ or $r'=0$ (see Figure~\ref{fig:gem-nodiamond}).
In each branch of the next branching rule, the sets $\{X_1, \dots, X_k\}$ will be modified into $\{X_1', \dots, X_k'\}$.\\

\begin{figure}
\begin{center}
\includegraphics[scale=0.7]{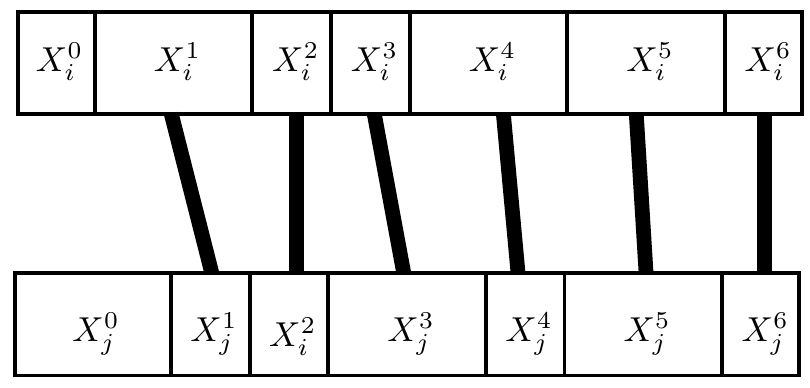}
\caption{Schema of the adjacencies between $X_i$ and $X_i$ when they do not contain a balanced diamond ($q=6$). An edge represent a complete relation between the corresponding subsets.}
\label{fig:gem-nodiamond}
\end{center}
\end{figure}

\textit{\textbf{Branching rule:} Let $i, j \in \{1, \dots, k\}$, $i \neq j$ such that $X_i \cup X_j$ has no balanced diamond. Then perform the following branching:
\begin{itemize}
	\item Branch 1: $X_i' = X_i^0$ and $X_z' = X_z$ for $z \in [k] \setminus \{i\}$
	\item Branch 2: $X_j' = X_j^0$ and $X_z' = X_z$ for $z \in [k] \setminus \{j\}$
	\item Branch 3: pick a set $T \subseteq \{1, \dots, q\}$ uniformly at random, then:
	\begin{itemize}
		\item $X_i' = \bigcup_{r \in T} X_i^r$
		\item $X_j' = \bigcup_{r \notin T} X_j^r$
		\item $X_z' = X_z$ for $z \in [k] \setminus \{i, j\}$
	\end{itemize}
\end{itemize}
}

Consider the graph $\mathcal{G}(X_1, \dots, X_k)$ having one vertex per set $X_i$, and an edge between $X_i$ and $X_j$ if these two sets are adjacent. We now prove the following:

\begin{claim}\label{claim:branchingdecrease}
The graph $\mathcal{G}(X_1', \dots, X_k')$ has one edge less than $\mathcal{G}(X_1, \dots, X_k)$
\end{claim}
\begin{claimproof}
In all three branches, observe that there is no edge between $X_i'$ and $X_j'$.
\end{claimproof}

\begin{claim}\label{claim:equivneg}
If $G$ has independent set of size $k$, then no graph obtained after the branching contains an independent set of size $k$.
\end{claim}
\begin{claimproof}
Observe that in all branches, $\bigcup_{z=1}^k X_z' \subseteq \bigcup_{z=1}^k X_z$, that is, each graph obtained in each branch is an induced subgraph of $G$.
\end{claimproof}

\begin{claim}\label{claim:equivproba}
If $G$ has a rainbow independent set, then with probability at least $\frac{1}{2}$, at least one branch leads to a graph having a rainbow independent set.
\end{claim}
\begin{claimproof}
Suppose that $G$ contains a rainbow independent set $S^*$. If $S^*$ intersects $X_i^0$, then $S^*$ also exists in the graph of the first branch. If $S^*$ intersects $X_j^0$, then $S^*$ also exists in the graph of the second branch.
The last case is where $S^*$ intersects $X_i^{r_1}$ and $X_j^{r_2}$, for some $r_1, r_2 \in \{1, \dots q\}$. In that case, there is a probability of $\frac{1}{2}$ that $r_1 \in T$ and $r_2 \notin T$, which concludes the proof of the claim
\end{claimproof}

We may now assume that the previous branching rule cannot apply. For the sake of readability, we keep the notation $X_1, \dots, X_k$ in order to denote our instance, even after an eventual application of the previous branching rule. For every $X_i$, $X_j$ with $i \neq j$, there is either (i) no edge between $X_i$ and $X_i$, or (ii) a balanced diamond induced by $X_i \cup X_j$.
Hence, Claim~\ref{claim:diamondtwins} implies that each connected component of the graph induced by $\bigcup_{i=1}^k X_i$ is a module with respect to the clique $\{c_1, \dots, c_{k-1}\}$. In particular, each connected component is dominated by some $c_p$, with $p \in \{1, \dots, k-1\}$, and is thus $P_4$-free (otherwise, a $P_4$ together with this vertex $c_p$ induce a gem), which means that we can decide in polynomial time whether $G$ contains an independent set of size $k$, by deciding the problem in every connected component separately.

By Claim~\ref{claim:branchingdecrease}, the previous branching rule can be applied at most ${k \choose 2}$ times. Hence, by Claim~\ref{claim:equivproba}, if $G$ contains a rainbow independent set, then our algorithm will find such a set with probability at least $\frac{1}{2^{k^2}}$. Finally, by Claim~\ref{claim:equivneg}, if $G$ does not contain any independent set of size $k$, then our algorithm will answer ``No''. This proves that our algorithm is a one-sided error Monte Carlo algorithm with success probability at least $\frac{1}{2^{k^2}}$, which can be turned into a randomized $FPT$ algorithm for \faugramsey{$f$}.
\end{proof}

\section{Polynomial (Turing) kernels}\label{sec:kernels}

In this section we investigate some special cases of Section~\ref{sec:clique-minus-bipartite}, in particular when $H$ is a clique of size $r$ minus a claw with $s$ branches, for $s < r$. Although Theorem~\ref{thm:clique-minus-bipartite} proves that \shortindepset is FPT for every possible values of $r$ and $s$, we show that when $s \ge r-2$, the problem admits a polynomial Turing kernel, while for $s \le 2$, it admits a polynomial kernel. Notice that the latter result is somehow tight, as Corollary~\ref{cor:nopolyker} shows that \shortindepset cannot admit a polynomial kernel in $(K_r \setminus K_{1, s})$-free graphs whenever $s \ge 3$.

\subsection{Positive results}\label{sec:kernels-positive}
The main ingredient of the two following results is a constructive version of the Erd\"os-Hajnal theorem for the concerned graph classes:

\begin{lemma}[Constructive Erd\"os-Hajnal for $K_r \setminus K_{1, s}$]\label{lem:erdos-hajnal}
For every $r \ge 2$ and $s < r$, there exists a polynomial-time algorithm which takes as input a connected $(K_r \setminus K_{1, s})$-free graph $G$, and construct either a clique or an independent set of size $n^{\frac{1}{r-1}}$, where $n$ is the number of vertices of $G$.
\end{lemma}
\begin{proof}
First consider the case $s=r-1$, \ie the forbidden graph is $K_{r-1}$ plus an isolated vertex. If $G$ contains a vertex $v$ with non-neighborhood $N$ of size at least $n^{\frac{r-2}{r-1}}$, then, since $G[N]$ is $K_{r-1}$-free, by Ramsey's theorem, it must contains an independent set of size $|N|^{\frac{1}{r-2}}=n^{\frac{1}{r-1}}$, which can be found in polynomial time. 
We may now assume that the maximum non-degree\footnote{The non-degree of a vertex is the size of its non-neighborhood.} of $G$ is $n^{\frac{r-2}{r-1}}-1$. We construct a clique $v_1$, $\dots$, $v_q$ in $G$ by picking an arbitrary vertex $v_1$, removing its non-neighborhood, then picking another vertex $v_2$, removing its non-neighborhood, and repeating this process until the graph becomes empty. Using the above argument on the maximum non-degree, this process can be applied $\frac{n}{n^{\frac{r-2}{r-1}}} = n^{\frac{1}{r-1}}$ times, corresponding to the size of the constructed clique.

Now, we make an induction on $r-1-s$ (the base case is above). If $G$ contains a vertex $v$ with neighborhood $N$ of size at least $n^{\frac{r-2}{r-1}}$, then, since $G[N]$ is $(K_{r-1} \setminus K_s)$-free, by induction it admits either a clique or an independent set of size $|N|^{\frac{1}{r-2}} = n^{\frac{1}{r-1}}$, which can be found in polynomial time. 
We may now assume that the maximum degree of $G$ is $n^{\frac{r-2}{r-1}}-1$. We construct an independent set $v_1$, $\dots$, $v_q$ in $G$ by picking an arbitrary vertex $v_1$, removing its neighborhood, and repeating this process until the graph becomes empty. Using the above argument on the maximum degree, this process can be applied $\frac{n}{n^{\frac{r-2}{r-1}}} = n^{\frac{1}{r-1}}$ times, corresponding to the size of the constructed independent set.
\end{proof}

\begin{theorem}\label{thm:r-minus-two}
$\forall r \ge 2$, \shortindepset in $(K_r \setminus K_{1, r-2})$-free graphs has a polynomial Turing kernel.
\end{theorem}
\begin{proof}
The problem is polynomial for $r=2$ and $r=3$, hence we suppose $r \ge 4$.
Suppose we have an algorithm $\mathcal{A}$ which, given a graph $J$ and an integer $i$ such that $|V(J)| = O(i^{r-1})$, decides whether $J$ has an independent set of size $i$ in constant time.
Having a polynomial algorithm for \shortindepset assuming the existence of $\mathcal{A}$ implies a polynomial Turing kernel for the problem~\cite{CyFoKoLoMaPiPiSa15}. 
To do so, we will present an algorithm $\mathcal{B}$ which, given a \emph{connected} graph $G$ and an integer $k$, outputs a polynomial (in $|V(G)|$) number of instances of size $O(k^{r-1})$, such that one of them is positive iff the former one is.
With this algorithm in hand, we obtain the polynomial Turing kernel as follows: let $G$ and $k$ be an instance of \shortindepset. Let $V_1$, $\dots$, $V_{\ell}$ be the connected components of $G$. For every $j \in \{1, \dots, \ell\}$, we determine the size of a maximum independent set $k_j$ of $G[V_j]$ by first invoking, for successive values $i=1, \dots, k$, the algorithm $\mathcal{B}$ on input $(G[V_j], i)$, and then $\mathcal{A}$ on each reduced instance. At the end of the algorithm, we answer $YES$ iff $\sum_{j=1}^{\ell} k_i \ge k$.

We now describe the algorithm $\mathcal{B}$. Let $(G, k)$ be an input, with $n=|V(G)|$.
By Lemma~\ref{lem:erdos-hajnal}, we start by constructing a clique $C$ of size at least $n^{\frac{1}{r-1}}$ in polynomial time. We assume that $|C| > r^2$, since otherwise the instance is already reduced.

Let $B = N(C)$. First observe that for every $u \in B$, $|N_C(u)| \ge |C|-(r-3)$. Indeed, if $|N_C(u)| \le |C|-(r-2)$, then the graph induced by $r-2$ non-neighbors of $u$ in $C$ together with $u$ and a neighbor of $u$ in $C$ (which exists since $|C| > r^2$) is isomorphic to $K_r \setminus K_{1, r-2}$.
Secondly, we claim that $V(G)=C \cup B$: for the sake of contradiction, take $v \in N(B) \setminus C$, and let $u \in B$ be such that $uv \in E(G)$. By the previous argument, $u$ has at least $|C|-r+3 \ge r-2$ neighbors in $C$ which, in addition to $u$ and $v$, induce a graph isomorphic to $K_r \setminus K_{1, r-2}$.

The algorithm outputs, for every $u \in B$, the graph induced by $B \setminus N[u]$, and, for every $u \in B$ and every $v \in C$ such that $uv \notin E(G)$, the graph induced by $B \setminus (N[u] \cup N[v])$. 
The correctness of the algorithm follows from the fact that if $G$ has an independent set $S$ of size $k > 1$, then either:
\begin{itemize}
	\item $S \cap C = \emptyset$, in which case $S$ lies entirely in $B \setminus N[u]$ for any $u \in S$, or 
	\item $S \cap C = \{v\}$ for some $v \in C$, in which case $S \setminus \{v\}$ lies entirely in $B \setminus (N[u] \cup N[v])$ for any $u \in S \cap B$.
\end{itemize}
We now argue that each of these instances has $O(k^{r-3})$ vertices. To do so, observe that for any $u \in B$, $B \setminus N[u]$ does not contain $K_{r-2}$ as an induced subgraph: indeed, since $|C| > r^2$, then any set of $r-2$ vertices of $B$ must have a common neighbor in $C$. Taking a clique of size $r-2$ in $B$ together with its common neighbor in $C$ and $u$ would induce a graph isomorphic to $K_r \setminus K_{1, r-2}$. 
Since each of these instances is $K_{r-2}$-free, applying Ramsey's theorem to each of them allows us to either construct an independent set of size $k-1$ in one of them (and thus output an independent set of size $k$ in $G$), or to prove that each of them has at most $O(k^{r-3})$ vertices.
At the end, this algorithm outputs $O(n^2)$ instances, each having $O(k^{r-3})$ vertices.
\end{proof}

Since a $(K_r \setminus K_{1, r-1})$-free graph is $(K_{r-1} \setminus K_{1, r-2})$-free, we have the following:
\begin{corollary}\label{cor:r-minus-one}
$\forall r \ge 2$, \shortindepset in $(K_r \setminus K_{1, r-1})$-free graphs has a polynomial Turing Kernel.
\end{corollary}

In other words, $(K_r \setminus K_{1, r-1})$ is a clique of size $r-1$ plus an isolated vertex. Observe that the previous corollary can actually be proved in a very simple way: informally, we can ``guess'' a vertex $v$ of the solution, and return its non-neighborhood together with parameter $k-1$. Since this non-neighborhood is $K_{r-1}$-free, it can be reduced to a $O(k^{r-2})$-sized instance. This is perhaps the most simple example of a problem admitting a polynomial Turing kernel but no polynomial kernel\footnote{\label{footnote:nopolyker}Unless $NP \subseteq coNP/poly$.} (as we will prove later in Theorem~\ref{thm:nopolyker}). By considering the complement of graphs, it implies the even simpler following observation: \textsc{Maximum Clique} has a $O(k^2)$ Turing kernel on \emph{claw}-free graphs, but no polynomial kernel\footnoteref{footnote:nopolyker}.

\begin{theorem}\label{thm:two}
$\forall r \ge 3$, \shortindepset in $(K_r \setminus K_{1, 2})$-free graphs has a kernel with $O(k^{r-1})$ vertices.
\end{theorem}
\begin{proof}
For $r = 3$, the problem is polynomial, so we assume $r \ge 4$.
The algorithm consists in constructing, by Lemma~\ref{lem:erdos-hajnal}, a clique $C$ of size at least $n^{\frac{1}{r-1}}$ in polynomial time. We present a reduction rule in the case $|C| > (k-1)(r-4)+1$. If this rule cannot apply, then it means that the number of vertices of the reduced instance is $O(k^{r-1})$.

First observe that for every $u \in N(C)$, then either $|N_C(u)|=|C|-1$, or $|N_C(u)| \le r-4$. Indeed, suppose that $r-3 \le |N_C(u)| \le |C|-2$. Then $u$ together with $r-3$ of its neighbors in $C$ and $2$ of its non-neighbors in $C$ induce a graph isomorphic to $K_r \setminus K_{1, 2}$, a contradiction.
Let $B = \{u \in N(C) : |N_C(u)| = |C|-1\}$ and $D = \{u \in N(C) : |N_C(u)| \le r-4\}$.

We claim that $C \cup B$ is a complete $|C|$-multipartite graph. To do so, we prove that for $u,v \in B$, $N_C(u)=N_C(v)$ implies $uv \notin E(G)$, and $N_C(u) \neq N_C(v)$ implies $uv \in E(G)$.
Suppose that $N_C(u) = N_C(v) = \{x\}$. If $uv \in E(G)$, then $u$, $v$, $x$ together with $r-3$ vertices of $C$ different from $x$ induce a graph isomorphic to $K_r \setminus K_{1, 2}$, which is impossible.
Suppose now that $N_C(u) = x_u \neq x_v = N_C(v)$. If $uv \notin E(G)$, then $u$, $v$, $x_u$ together with $r-3$ vertices of $C$ different from $x_u$ and $x_v$ induce a graph isomorphic to $K_r \setminus K_{1, 2}$, which is impossible.

Thus, we now write $C \cup B = S_1 \cup \dots \cup S_{|C|}$, 
where, for every $i,j \in \{1, \dots, |C|\}$, $i \neq j$, $S_i$ induces an independent set, and $S_i \cup S_j$ induces a complete bipartite graph. We assume $|S_1| \ge |S_2| \ge \dots \ge |S_{|C|}|$. Recall that $|C| > (k-1)(r-4)+1$.
Using the same arguments as previously, we can show that every vertex of $D$ is adjacent to at most $r-4$ different parts among $C \cup B$. More formally: for every $u \in D$, we have $|\{S_i : N(u) \cap S_i \neq \emptyset\}| \le r-4$. Let $q = (k-1)(r-4)+1$.
The reduction consists in removing $S_{q+1} \cup \dots \cup S_{|C|}$. Clearly it runs in polynomial time.

Let $G'$ denote the reduced instance. Obviously, if $G'$ has an independent set of size $k$, then $G$ does, since $G'$ is an induced subgraph of $G$. It remains to show that the converse is also true.
Let $X$ be an independent set of $G$ of size $k$. If $X \cap \left( \cup_{i=q+1}^{|C|} S_i \right) = \emptyset$, then $X$ is also an independent set of size $k$ in $G'$, thus we suppose $X \cap \left( \cup_{i=q+1}^{|C|} S_i \right) = X_r \neq \emptyset$. In particular, since $C \cup B$ is a multipartite graph, there is a unique $i \in \{1, \dots, |C|\}$ such that $X \cap S_i \neq \emptyset$, and $i \ge q+1$. Since every vertex of $D$ is adjacent to at most $r-4$ parts of $C \cup B$, and since $q =(k-1)(r-4)+1$, there must exist $j \in \{1, \dots, q\}$ such that $N(X \cap D) \cap S_j = \emptyset$. Moreover, $|S_j| \ge |S_i|$. Hence, $(X \setminus S_i) \cup S_j$ is an independent set of size at least $k$ in $G'$.
\end{proof}

Observe that a $(K_r \setminus K_{2})$-free graph is $(K_{r+1} \setminus K_{1, 2})$-free, hence we have the following, which answers a question of \cite{Dabrowski12}.
\begin{corollary}\label{cor:one}
$\forall r \ge 1$, \shortindepset in $(K_r \setminus K_{2})$-free graphs has a kernel with $O(k^{r-1})$ vertices.
\end{corollary}

\subsection{Kernel lower bounds}\label{sec:kernel-lower-bound}
\begin{definition}
 Given the graphs $H$, $H_1$, $\dots$, $H_p$, we say that $(H_1, \dots, H_p)$ is a multipartite decomposition of $H$ if $H$ is isomorphic to $H_1 + \dots + H_p$. We say that $(H_1, \dots, H_p)$ is maximal if, for every multipartite decomposition $(H_1', \dots, H_q')$ of $H$, we have $p > q$.
\end{definition}

It can easily be seen that for every graph $H$, a maximal multipartite decomposition of $H$ is unique. We have the following:

\begin{theorem}\label{thm:nopolyker}
 Let $H$ be any fixed graph, and let $H = H_1 + \dots + H_p$ be the maximal multipartite decomposition of $H$. If, for some $i \in [p]$, \shortindepset is NP-hard in $H_i$-free graphs, then \shortindepset does not admit a polynomial kernel in $H$-free graphs unless NP $\subseteq$ coNP/poly.
\end{theorem}
\begin{proof}
 We construct an OR-cross-composition from \shortindepset in $H_i$-free graphs. For more details about cross-compositions, see~\cite{BoJaKr14}. Let $G_1, \dots, G_t$ be a sequence of $H_i$-free graphs, and let $G' = G_1 + \dots + G_t$. Then we have the following:
 \begin{itemize}
  \item $\alpha(G') = \max_{i=1...t} \alpha(G_i)$, since, by construction of $G'$, any independent set cannot intersect the vertex set of two distinct graphs $G_i$ and $G_j$.
  \item $G'$ is $H$-free. Indeed, suppose that $X \subseteq V(G')$ induces a graph isomorphic to $H$, and let $X_j = X \cap V(G_j)$ for every $j \in [p]$. Then observe that the graphs induced by the non-empty sets $X_j$ form a multipartite decomposition of $H$, and thus there must exist $j \in [p]$ such that $G_j[X_j]$ contains $H_i$ as an induced subgraph, a contradiction.
 \end{itemize}
 These two arguments imply a cross-composition from \shortindepset in $H_i$-free graphs to \shortindepset in $H$-free graphs.
\end{proof}

The next results shows that the polynomial kernel obtained in the previous section for $(K_r \setminus K_{1, s})$-free graphs, $s \le 2$, is somehow tight.

\begin{corollary}\label{cor:nopolyker}
 For $r \ge 4$, and every $3 \le s \le r-1$, \shortindepset in $(K_r \setminus K_{1, s})$-free graphs does not admit a polynomial kernel unless NP $\subseteq$ coNP/poly.
\end{corollary}
\begin{proof}
 In that case, observe that the maximal multipartite decomposition of $K_r \setminus K_{1, s}$ is $$\dot{K_{s}} + \overbrace{K_1 + \dots + K_1}^{r-1-s \text{ times }}$$
where $\dot{K_{s}}$ denotes the clique of size $s$ plus an isolated vertex. Moreover, \shortindepset is NP-hard in $\dot{K_{s}}$-free graphs for $s \ge 3$.
\end{proof}

%In Section~\ref{sec:kernels}, we prove that a polynomial kernel exists in $(K_r \setminus K_{1,s})$-free graphs whenever $s \le 2$.

We conjecture that Theorem~\ref{thm:nopolyker} actually captures all possible negative cases concerning the kernelization of the problem. Informally speaking, our intuition is the natural idea that the join operation between graphs seems the only way to obtain $\alpha(G) = O(\max_{i=1, \dots, t} \alpha(G_i))$, which is the main ingredient of OR-compositions.

\begin{conjecture}
Let $H$ be any fixed graph, and $H=H_1 + \dots + H_p$ be its maximal multipartite decomposition. Then, assuming that $NP \not \subseteq coNP/poly$, \shortindepset admits a polynomial kernel in $H$-free graphs if and only if it is polynomial in $H_i$-free graph, for every $i \in [p]$.
\end{conjecture}

\section{Conclusion and open problems}

We started to unravel the FPT/$W[1]$-hard dichotomy for \shortindepset in $H$-free graphs, for a fixed graph $H$.
%We made some significant progress on the complexity side.
%The parameterized complexity of \shortindepset was unknown on $C_4$-free graphs.
At the cost of one reduction, we showed that it is $W[1]$-hard as soon as $H$ is not chordal, even if we simultaneously forbid induced $K_{1,4}$ and trees with at least two branching vertices.
Tuning this construction, it is also possible to show that if a connected $H$ is not roughly a "path of cliques" or a "subdivided claw of cliques", then \shortindepset is $W[1]$-hard.
More formally, with the definitions of Section~\ref{sec:harddecompositions}, the remaining connected open cases are when $H$ has an almost strong clique decomposition on a subdivided claw or a nearly strong clique decomposition on a path.
In this language, we showed that for every connected graph $H$ with a strong clique decomposition on a $P_3$, there is an FPT algorithm.
However, we also proved that for a very simple graph $H$ with a strong clique decomposition on the claw, \shortindepset is $W[1]$-hard.
This suggests that the FPT/$W[1]$-hard dichotomy will be somewhat subtle.
For instance, easy cases for the parameterized complexity do \emph{not} coincide with easy cases for the classical complexity where each vertex can be blown into a clique. 
For graphs $H$ with a clique decomposition on a path, the first unsolved cases are $H$ having:
\begin{itemize}
\item an almost strong clique decomposition on $P_3$;
\item a nearly strong clique decomposition on $P_3$;
\item a strong clique decomposition on $P_4$.
\end{itemize}
For graphs $H$ with a clique decomposition on the claw, an interesting open question is the case of \emph{cricket}-free graphs ($T_{1,1,2}$-free with our notation defined before Theorem~\ref{thm:T122}), and, more generally, in $T_{1,1,s}$-free graphs.

For disconnected graphs $H$, we obtained an FPT algorithm when $H$ is a cluster (\ie, a disjoint union of cliques).
We conjecture that, more generally, the disjoint union of two easy cases is an easy case; formally, \emph{if \shortindepset is FPT in $G$-free graphs and in $H$-free graphs, then it is FPT in $G \uplus H$-free graphs}.

A natural question regarding our two $FPT$ algorithms of Section~\ref{sec:positive-two} concerns the existence of polynomial kernels. In particular, we even do not know whether the problem admits a kernel for very simple cases, such as when $H=K_5 \setminus K_3$ or $H=K_5 \setminus K_{2, 2}$.

A more anecdotal conclusion is the fact that the parameterized complexity of the problem on $H$-free graphs is now complete for every graph $H$ on four vertices, including concerning the polynomial kernel question (see Figure~\ref{table:fourvertices}).
%, whereas the $FPT$/$W[1]$-hard question remains open for only five graphs $H$ on five vertices (see Figure~\ref{fig:remainingfive}).

\begin{figure}
{\setlength{\extrarowheight}{15pt}
\begin{tabular}{|c|p{1.5cm}|p{1.5cm}|p{1.5cm}|p{1.5cm}|}
\hline
Graph & P & PK & PTK & FPT \\[5ex]
  \hline
  %EDGELESS
 \begin{tikzpicture} [scale=0.5]
 \node[draw, circle, fill=black, inner sep=-0.05cm] at (0,0) {} ;
 \node[draw, circle, fill=black, inner sep=-0.05cm] at (0,1) {} ;
 \node[draw, circle, fill=black, inner sep=-0.05cm] at (1,0) {} ;
 \node[draw, circle, fill=black, inner sep=-0.05cm] at (1,1) {} ;
 \end{tikzpicture}  & \cellcolor{green!20} Obvious & \cellcolor{green!20} & \cellcolor{green!20} & \cellcolor{green!20} \\
  \hline
%O_p
 \begin{tikzpicture} [scale=0.5]
 \node[draw, circle, fill=black, inner sep=-0.05cm] at (0,0) {} ;
 \node[draw, circle, fill=black, inner sep=-0.05cm] at (0,1) {} ;
 \node[draw, circle, fill=black, inner sep=-0.05cm] at (1,0) {} ;
 \node[draw, circle, fill=black, inner sep=-0.05cm] at (1,1) {} ;
 \draw (0,0) -- (1,0);
 \end{tikzpicture}  & \cellcolor{green!20} Obvious & \cellcolor{green!20} & \cellcolor{green!20} & \cellcolor{green!20} \\
  \hline
  
  %P2+v
 \begin{tikzpicture} [scale=0.5]
 \node[draw, circle, fill=black, inner sep=-0.05cm] at (0,0) {} ;
 \node[draw, circle, fill=black, inner sep=-0.05cm] at (0,1) {} ;
 \node[draw, circle, fill=black, inner sep=-0.05cm] at (1,0) {} ;
 \node[draw, circle, fill=black, inner sep=-0.05cm] at (1,1) {} ;
 \draw (0,0) -- (1,0);
 \draw (0,0) -- (0,1);
 \end{tikzpicture}  & \cellcolor{green!20} Obvious & \cellcolor{green!20} & \cellcolor{green!20} & \cellcolor{green!20} \\
  \hline
  
%2K2
 \begin{tikzpicture} [scale=0.5]
 \node[draw, circle, fill=black, inner sep=-0.05cm] at (0,0) {} ;
 \node[draw, circle, fill=black, inner sep=-0.05cm] at (0,1) {} ;
 \node[draw, circle, fill=black, inner sep=-0.05cm] at (1,0) {} ;
 \node[draw, circle, fill=black, inner sep=-0.05cm] at (1,1) {} ;
 \draw (0,0) -- (1,0);
 \draw (0,1) -- (1,1);
 \end{tikzpicture}  & \cellcolor{green!20} \cite{Ale91} & \cellcolor{green!20} & \cellcolor{green!20} & \cellcolor{green!20} \\
  \hline

%claw
 \begin{tikzpicture} [scale=0.5]
 \node[draw, circle, fill=black, inner sep=-0.05cm] at (0,0) {} ;
 \node[draw, circle, fill=black, inner sep=-0.05cm] at (0,1) {} ;
 \node[draw, circle, fill=black, inner sep=-0.05cm] at (1,0) {} ;
 \node[draw, circle, fill=black, inner sep=-0.05cm] at (1,1) {} ;
 \draw (0,0) -- (1,0);
 \draw (0,0) -- (0,1);
 \draw (0,0) -- (1,1);
 \end{tikzpicture}  & \cellcolor{green!20} \cite{Mi80}  & \cellcolor{green!20} & \cellcolor{green!20} & \cellcolor{green!20} \\
  \hline

%P4
 \begin{tikzpicture} [scale=0.5]
 \node[draw, circle, fill=black, inner sep=-0.05cm] at (0,0) {} ;
 \node[draw, circle, fill=black, inner sep=-0.05cm] at (0,1) {} ;
 \node[draw, circle, fill=black, inner sep=-0.05cm] at (1,0) {} ;
 \node[draw, circle, fill=black, inner sep=-0.05cm] at (1,1) {} ;
 \draw (0,0) -- (1,0);
 \draw (0,0) -- (0,1);
 \draw (0,1) -- (1,1);
 \end{tikzpicture}  & \cellcolor{green!20} \cite{CoPeSt85} & \cellcolor{green!20} & \cellcolor{green!20} & \cellcolor{green!20} \\
  \hline
%K4
 \begin{tikzpicture} [scale=0.5]
 \node[draw, circle, fill=black, inner sep=-0.05cm] at (0,0) {} ;
 \node[draw, circle, fill=black, inner sep=-0.05cm] at (0,1) {} ;
 \node[draw, circle, fill=black, inner sep=-0.05cm] at (1,0) {} ;
 \node[draw, circle, fill=black, inner sep=-0.05cm] at (1,1) {} ;
 \draw (0,0) -- (1,0);
 \draw (0,0) -- (1,1);
 \draw (0,0) -- (0,1);
 \draw (1,0) -- (0,1);
 \draw (1,0) -- (1,1);
 \draw (0,1) -- (1,1);
 \end{tikzpicture}  & \cellcolor{red!20} Thm.~\ref{thm:nph} & \cellcolor{green!20} Ramsey & \cellcolor{green!20}  & \cellcolor{green!20} \\
  \hline
%diamond
 \begin{tikzpicture} [scale=0.5]
 \node[draw, circle, fill=black, inner sep=-0.05cm] at (0,0) {} ;
 \node[draw, circle, fill=black, inner sep=-0.05cm] at (0,1) {} ;
 \node[draw, circle, fill=black, inner sep=-0.05cm] at (1,0) {} ;
 \node[draw, circle, fill=black, inner sep=-0.05cm] at (1,1) {} ;
 \draw (0,0) -- (1,0);
 \draw (0,0) -- (1,1);
 \draw (0,0) -- (0,1);
 \draw (1,0) -- (1,1);
 \draw (0,1) -- (1,1);
 \end{tikzpicture}  & \cellcolor{red!20} Thm.~\ref{thm:nph}  & \cellcolor{green!20} Cor.~\ref{cor:one} & \cellcolor{green!20} & \cellcolor{green!20} \\
  \hline
%flag
 \begin{tikzpicture} [scale=0.5]
 \node[draw, circle, fill=black, inner sep=-0.05cm] at (0,0) {} ;
 \node[draw, circle, fill=black, inner sep=-0.05cm] at (0,1) {} ;
 \node[draw, circle, fill=black, inner sep=-0.05cm] at (1,0) {} ;
 \node[draw, circle, fill=black, inner sep=-0.05cm] at (1,1) {} ;
 \draw (0,0) -- (1,0);
 \draw (0,0) -- (1,1);
 \draw (0,0) -- (0,1);
 \draw (0,1) -- (1,1);
 \end{tikzpicture}  & \cellcolor{red!20} Thm.~\ref{thm:nph}  & \cellcolor{green!20} Thm.~\ref{thm:two} & \cellcolor{green!20} & \cellcolor{green!20} \\
  \hline
%K3+v
 \begin{tikzpicture} [scale=0.5]
 \node[draw, circle, fill=black, inner sep=-0.05cm] at (0,0) {} ;
 \node[draw, circle, fill=black, inner sep=-0.05cm] at (0,1) {} ;
 \node[draw, circle, fill=black, inner sep=-0.05cm] at (1,0) {} ;
 \node[draw, circle, fill=black, inner sep=-0.05cm] at (1,1) {} ;
 \draw (0,0) -- (1,0);
 \draw (0,0) -- (0,1);
 \draw (1,0) -- (0,1);
 \end{tikzpicture}  & \cellcolor{red!20} & \cellcolor{red!20} Cor.~ \ref{cor:nopolyker} & \cellcolor{green!20} Cor.~\ref{cor:r-minus-one} & \cellcolor{green!20} \\
  \hline
%C4
 \begin{tikzpicture} [scale=0.5]
 \node[draw, circle, fill=black, inner sep=-0.05cm] at (0,0) {} ;
 \node[draw, circle, fill=black, inner sep=-0.05cm] at (0,1) {} ;
 \node[draw, circle, fill=black, inner sep=-0.05cm] at (1,0) {} ;
 \node[draw, circle, fill=black, inner sep=-0.05cm] at (1,1) {} ;
 \draw (0,0) -- (1,0);
 \draw (1,0) -- (1,1);
 \draw (1,1) -- (0,1);
 \draw (0,1) -- (0,0);

 \end{tikzpicture}  & \cellcolor{red!20} & \cellcolor{red!20} & \cellcolor{red!20} & \cellcolor{red!20} Thm.~\ref{thm:Whard} \\
  \hline
\end{tabular}
}
\caption{Status of the problem for graphs $H$ on four vertices. $P$, $PK$, $PTK$ respectively stand for \emph{Polynomial}, \emph{$NP$-hard but admits a polynomial kernel}, and \emph{no polynomial kernel unless $NP \subseteq coNP/poly$ but admits a polynomial Turing kernel}.}
\label{table:fourvertices}
\end{figure}

\bibliography{biblio}

\end{document}